\documentclass[10pt,a4paper]{article}

\usepackage[utf8]{inputenc} 
\usepackage[english]{babel} 

\usepackage{graphicx}
\usepackage{mathtools, cuted}
\usepackage{hyperref}

\usepackage{mathtools}

\usepackage{amsmath, amssymb}
\usepackage{amsthm}
\usepackage{amsmath}

\usepackage{appendix}

\usepackage{setspace}

\usepackage{amssymb}
\usepackage[a4paper, total={6in, 10in}]{geometry}

\usepackage{wrapfig}

\usepackage{chngcntr}
\usepackage{apptools}

\AtAppendix{\counterwithin{lemma}{section}}

\usepackage{color}

\usepackage{anyfontsize}
\usepackage{amsthm}

\newtheorem{theorem}{Theorem}[section]
\newtheorem{definition}[theorem]{Definition}
\newtheorem{corollary}[theorem]{Corollary}
\newtheorem{lemma}[theorem]{Lemma} 
\newtheorem{proposition}[theorem]{Proposition} 
\newtheorem{remark}[theorem]{Remark}

\date{}

\begin{document}

\title{Model Design and Representations of CM Sequences}
\author{Reza Rezaie and X. Rong Li
\thanks{The authors are with the Department of Electrical Engineering, University of New Orleans, New Orleans, LA 70148. Email addresses are
 {\tt\small rrezaie@uno.edu} and {\tt\small xli@uno.edu}. Research was supported by NASA through grant NNX13AD29A.}}

\maketitle

\begin{abstract}
Conditionally Markov (CM) sequences are powerful mathematical tools for modeling problems. One class of CM sequences is the reciprocal sequence. In application, we need not only CM dynamic models, but also know how to design model parameters. Models of two important classes of nonsingular Gaussian (NG) CM sequences, called $CM_L$ and $CM_F$ models, and a model of the NG reciprocal sequence, called reciprocal $CM_L$ model, were presented in our previous works and their applications were discussed. In this paper, these models are studied in more detail, in particular their parameter design. It is shown that every reciprocal $CM_L$ model can be induced by a Markov model. Then, parameters of each reciprocal $CM_L$ model can be obtained from those of the Markov model. Also, it is shown that a NG $CM_L$ ($CM_F$) sequence can be represented by a sum of a NG Markov sequence and an uncorrelated NG vector. This (necessary and sufficient) representation provides a basis for designing parameters of a $CM_L$ ($CM_F$) model. From the CM viewpoint, a representation is also obtained for NG reciprocal sequences. This representation is simple and reveals an important property of reciprocal sequences. As a result, the significance of studying reciprocal sequences from the CM viewpoint is demonstrated. A full spectrum of dynamic models from a $CM_L$ model to a reciprocal $CM_L$ model is also presented.  

\end{abstract}

\textbf {Keywords:} Conditionally Markov, reciprocal, Markov, Gaussian, dynamic model, characterization. %\footnote{Research supported by NASA/LEQSF Phase03-06 through grant NNX13AD29A.} 

\section{Introduction}

Consider stochastic sequences defined over $[0,N]=\lbrace 0,1,\ldots,N \rbrace$. For convenience, let the index be time. A sequence is Markov if and only if (iff) conditioned on the state at any time $k$, the segment before $k$ is independent of the segment after $k$. A sequence is reciprocal iff conditioned on the states at any two times $j$ and $l$, the segment inside the interval $(j,l)$ is independent of the two segments outside $[j,l]$. As defined in \cite{CM_Part_I_Conf}, a sequence is $CM_F$ ($CM_L$) over $[k_1,k_2]$ iff conditioned on the state at time $k_1$ ($k_2$), the sequence is Markov over $[k_1+1,k_2]$ ($[k_1,k_2-1]$). The Markov sequence is a special case of the reciprocal sequence (i.e., each Markov sequence is a reciprocal sequence, but not vice versa) and the reciprocal sequence is a special case of the CM sequence (i.e., each reciprocal sequence is a CM sequence, but not vice versa).

Markov processes have been widely studied and used. However, they are not general enough for some problems \cite{Fanas1}--\cite{DW_Conf}, and more general processes are needed. Reciprocal processes are one generalization of Markov processes. The CM process (including the reciprocal process as a special case) provides a systematic and convenient generalization of the Markov process (based on conditioning) leading to various classes of processes \cite{CM_Part_I_Conf}.

Being a motivation for defining reciprocal processes \cite{Bernstein}, the problem posed by E. Schrodinger \cite{Schrodinger_1} about some behavior of particles can be studied in the reciprocal process setting. In \cite{Levy_2} reciprocal processes were discussed in the context of stochastic mechanics. In a quantized state space, finite-state reciprocal sequences were used in \cite{Fanas1}--\cite{White_Tracking1} for detection of anomalous trajectory patterns, intent inference, and tracking. The approach presented in \cite{Simon}--\cite{Simon2} for intent inference in an intelligent interactive vehicle's display is implicitly based on the idea of reciprocal processes. In \cite{Krener1}, the relation between acausal systems and reciprocal processes was discussed. Applications of reciprocal processes in image processing can be found in \cite{Picci}--\cite{Picci2}. Some CM sequences were used in \cite{DD_Conf}--\cite{DW_Conf} for trajectory modeling and prediction.

Gaussian CM processes were introduced in \cite{Mehr} based on mean and covariance functions, where the processes were assumed nonsingular on the interior of the index (time) interval. \cite{Mehr} considered conditioning at the first time of the CM interval. \cite{ABRAHAM} extended the definition of Gaussian CM processes (presented in \cite{Mehr}) to the general (Gaussian/non-Gaussian) case. In \cite{CM_Part_I_Conf} we presented definitions of different (Gaussian/non-Gaussian) CM processes based on conditioning at the first or the last time of the CM interval, studied (stationary/non-stationary) NG CM sequences, and presented their dynamic models and characterizations. Two of these models for two important classes of NG CM sequences (i.e., sequences being $CM_L$ or $CM_F$ over $[0,N]$) are called $CM_L$ and $CM_F$ models. Applications of CM sequences for trajectory modeling in different scenarios were also discussed. In addition, \cite{CM_Part_I_Conf} provided a foundation and preliminaries for studying the reciprocal sequence from the viewpoint of the CM sequence in \cite{CM_Part_II_A_Conf}.  

Reciprocal processes were introduced in \cite{Bernstein} and studied in \cite{Jamison_Reciprocal}--\cite{White_3} and others. \cite{Jamison_Reciprocal}--\cite{Roally} studied reciprocal processes in a general setting. \cite{ABRAHAM} made an inspiring comment that reciprocal and CM processes are related, and discussed the relationship between the Gaussian reciprocal process and the Gaussian CM process. \cite{CM_Part_II_A_Conf} elaborated on the comment of \cite{ABRAHAM} and obtained a relationship between (Gaussian/non-Gaussian) CM and reciprocal processes. It was shown in \cite{ABRAHAM} that a NG continuous-time CM (including reciprocal) process can be represented in terms of a Wiener process and an uncorrelated NG vector. Following \cite{ABRAHAM}, \cite{Carm}--\cite{Carm2} obtained some results about continuous-time Gaussian reciprocal processes. \cite{Krener_2}--\cite{Levy_Dynamic} presented state evolution models of Gaussian reciprocal processes. In \cite{Levy_Dynamic}, a dynamic model and a characterization of the NG reciprocal sequence were presented. It was shown that the evolution of a reciprocal sequence can be described by a second-order nearest-neighbor model driven by locally correlated dynamic noise \cite{Levy_Dynamic}. That model is a natural generalization of the Markov model. Due to the dynamic noise correlation and the nearest-neighbor structure, the corresponding state estimation is not straightforward. Recursive estimation of the sequence based on the model presented in \cite{Levy_Dynamic} was discussed in \cite{Bacca1}--\cite{Moura2}. A covariance extension problem for reciprocal sequences was addressed in \cite{Carli}. Modeling and estimation of finite-state reciprocal sequences were discussed in \cite{White}--\cite{White_3}. Based on the results of \cite{CM_Part_I_Conf}, in \cite{CM_Part_II_A_Conf} reciprocal sequences were studied from the CM viewpoint leading to simple and revealing results.

A typical application of reciprocal and CM sequences is in trajectory modeling and prediction with an intent or destination. One group of papers focuses on trajectory prediction without explicitly modeling trajectories. They use estimation approaches developed for the case of no intent/destination. Then, they utilize intent/destination information to improve the trajectory prediction performance (e.g., \cite{Hwang0}--\cite{Krozel}). The underlying trajectory model is not clear in such approaches. However, to study and generate trajectories, and analyze problems, it is desired to have a clear model. A rigorous mathematical model of the trajectories provides a solid basis for systematically handling relevant problems. Another group of papers tries to explicitly model trajectories. Due to many sources of uncertainty, trajectories are mathematically modeled as some stochastic processes (e.g., \cite{Fanas1}--\cite{White_Tracking1}). After quantizing the state space, \cite{Fanas1}--\cite{White_Tracking1} used finite-state reciprocal sequences for intent inference and trajectory modeling with destination/waypoint information. Reciprocal sequences provide an interesting mathematical tool for motion problems with destination information. However, it is not always feasible or efficient to quantize the state space. So, it is desirable to use continuous-state reciprocal sequences to model such trajectories. Gaussian sequences have a continuous-state space. A dynamic model of NG reciprocal sequences was presented in \cite{Levy_Dynamic}, which is the most significant paper on Gaussian reciprocal sequences. However, as mentioned above, due to the nearest neighbor structure and the colored dynamic noise, the model of \cite{Levy_Dynamic} is not easy to apply for trajectory modeling and prediction. For example, in the model of \cite{Levy_Dynamic}, the current state depends on the previous state and the next state. As a result, for estimation of the current state, prior information (density) of the next state is required. However, such information is not available.

We presented a different dynamic model of NG reciprocal sequence (called reciprocal $CM_L$ model) in \cite{CM_Part_II_A_Conf} from the CM viewpoint. That model has a good structure for trajectory modeling with a destination. More specifically, its structure can naturally model a destination. Also, recursive estimation based on the model of \cite{CM_Part_II_A_Conf} is straightforward. Like any model-based approach, to use it in application, we need to design its model parameters. 

In this paper, we present a rigorous and systematic approach for parameter design of a reciprocal $CM_L$ model, which is directly applicable to trajectory modeling with a destination. Following \cite{Jamison_Reciprocal}, \cite{Fanas2} obtained a transition probability function of a finite-state reciprocal sequence from a transition probability function of a finite-state Markov sequence in a quantized state space for a problem of intent inference and trajectory modeling. However, \cite{Fanas2} did not discuss if all reciprocal transition probability functions can be obtained from a Markov transition probability function, which is critical for the application considered in \cite{Fanas2}. In this paper, we make this issue clear based on our reciprocal $CM_L$ model. \color{black}\cite{Simon}--\cite{Simon2} obtained a transition density of a Gaussian bridging distribution from a Markov transition density. However, \cite{Simon}--\cite{Simon2} did not show what type of stochastic process was obtained for modeling their problem of intent inference. In other words, \cite{Simon}--\cite{Simon2} did not discuss what type of transition density was obtained. In this paper, we address this issue and make it clear. Including reciprocal sequences as a special case (\cite{ABRAHAM}, \cite{CM_Part_II_A_Conf}), CM sequences are more general for trajectory modeling with waypoints and/or a destination \cite{CM_Part_I_Conf}, for example, $CM_L$ sequences for trajectory modeling with destination information. However, guidelines for parameter design of a $CM_L$ model are lacking. In this paper, application of $CM_L$ sequences to trajectory modeling is discussed and guidelines for parameter design of $CM_L$ models are presented. Some classes of CM sequences provide models for more complicated trajectories. For example, a $CM_L \cap [0,k_2]$-$CM_L$ sequence (i.e., a sequence which is $CM_L$ over both $[0,N]$ and $[0,k_2]$) can be applied to modeling trajectories with a waypoint and a destination. However, a dynamic model of $CM_L \cap [0,k_2]$-$CM_L$ sequences is not available in the literature. We discuss such application and present a dynamic model for these CM sequences. Systematic modeling of trajectories in the above scenarios is desired but challenging. Different classes of CM sequences make it possible to achieve this goal. Then, for application of these CM sequences, we need to have their dynamic models and design their parameters. This is a main topic of this paper.

The main goal of this paper is three-fold: 1) to present approaches/guidelines for parameter design of $CM_L$, $CM_F$, and reciprocal $CM_L$ models in general and their application in trajectory modeling with destination in particular, 2) to obtain a representation of NG $CM_L$, $CM_F$, and reciprocal sequences, revealing a key fact behind these sequences, and to demonstrate the significance of studying reciprocal sequences from the CM viewpoint, and 3) to present a full spectrum of dynamic models from a $CM_L$ model to a reciprocal $CM_L$ model and show how models of various intersections of CM classes can be obtained.

The main contributions of this paper are as follows. From the CM viewpoint, we not only show how a Markov model induces a reciprocal $CM_L$ model, but also prove that \textit{every} reciprocal $CM_L$ model can be induced by a Markov model. Then, we give formulas to obtain parameters of the reciprocal $CM_L$ model from those of the Markov model. This approach is more intuitive than a direct parameter design of a reciprocal $CM_L$ model, because one usually has a much better intuitive understanding of Markov models. This is particularly useful for parameter design of a reciprocal $CM_L$ model for trajectory modeling with destination. In addition, our results make it clear that the transition density obtained in \cite{Simon}--\cite{Simon2} is actually a reciprocal transition density. A full spectrum of dynamic models from a $CM_L$ model to a reciprocal $CM_L$ model is presented. This spectrum helps to understand the gradual change from a $CM_L$ model to a reciprocal $CM_L$ model. Also, it is demonstrated how dynamic models for intersections of NG CM sequences can be obtained. In addition to their usefulness for application (e.g., application of $CM_L \cap [0,k_2]$-$CM_L$ sequences in trajectory modeling with a waypoint and a destination), these models are particularly useful to describe the evolution of a sequence (e.g., a reciprocal sequence) in more than one CM class. Based on a valuable observation, \cite{ABRAHAM} discussed representations of NG continuous-time CM (including reciprocal) processes in terms of a Wiener process and an uncorrelated NG vector. First, we show that the representation presented in \cite{ABRAHAM} is not sufficient for a Gaussian process to be reciprocal (although \cite{ABRAHAM} stated that it was sufficient, which has not been challenged or corrected so far). Then, we present a simple (necessary and sufficient) representation for NG reciprocal sequences from the CM viewpoint. This demonstrates the significance of studying reciprocal sequences from the CM viewpoint. Second, inspired by \cite{ABRAHAM}, we show that a NG $CM_L$ ($CM_F$) sequence can be represented by a NG Markov sequence plus an uncorrelated NG vector. This (necessary and sufficient) representation makes a key fact of CM sequences clear and is very helpful for parameter design of $CM_L$ and $CM_F$ models from a Markov model plus an uncorrelated NG vector. Third, we study the obtained representations of NG $CM_L$, $CM_F$, and reciprocal sequences and, as a by-product, obtain new representations of some matrices, which characterize NG $CM_L$, $CM_F$, and reciprocal sequences. 

A preliminary conference version of this paper is \cite{CM_Part_II_B_Conf}, where results were presented without proof. In this paper, we present all proofs and detailed discussion. Other significant results beyond \cite{CM_Part_II_B_Conf} include the following. The notion of a $CM_L$ model \textit{induced} by a Markov model is defined and such a model is studied in Subsection \ref{S_M_R} (Definition \ref{CML_Derived}, Corollary \ref{Reciprocal_Derived_Markov}, and Lemma \ref{Markov_In_CML_Reciprocal_Class}). Dynamic models are obtained for intersections of CM classes (Proposition \ref{CML_k1N_CMF_Dynamic} and Proposition \ref{CML_0k2_CML_Dynamic_2nd}). Uniqueness of the representation of a $CM_L$ ($CM_F$) sequence (as a sum of a NG Markov sequence and an uncorrelated NG vector) is proved (Corollary \ref{unique_rep}). Such a representation is also presented for reciprocal sequences (Proposition \ref{Reciprocal_Markov_z} and Proposition \ref{CML_z_Reciprocal}). Due to its usefulness for application, the notion of a $CM_L$ model \textit{constructed} from a Markov model is introduced and is compared with that of a $CM_L$ model \textit{induced} by a Markov model (Section \ref{Section_Rep}). As a by-product, representations of some matrices (that characterize CM sequences) are obtained in Corollary \ref{CML_Decomp} and \ref{Reciprocal_Decomp}.

The paper is organized as follows. Section \ref{Section_Definition_Preliminary} reviews some definitions and results required for later sections. In Section \ref{Section_Model}, a reciprocal $CM_L$ model and its parameter design are discussed. Also, it is shown how dynamic models for intersections of CM classes can be obtained. In Section \ref{Section_Rep}, a representation of NG $CM_L$ ($CM_F$) sequences are presented and parameter design of $CM_L$ and $CM_F$ models is discussed. Section \ref{Summary_Conclusions} contains a summary and conclusions.

\section{Definitions and Preliminaries}\label{Section_Definition_Preliminary}

%\subsection{Conventions}\label{Subsection_Convention}

We consider stochastic sequences defined over the interval $[0,N]$, which is a general discrete index interval. For convenience this discrete index is called time. Also, we consider:
\begin{align*}
[i,j]& \triangleq \lbrace i,i+1,\ldots ,j-1,j \rbrace, \quad i<j \\
[x_k]_{i}^{j} & \triangleq \lbrace x_k, k \in [i,j] \rbrace\\
[x_k] & \triangleq [x_k]_{0}^{N}\\
i,j,l, k_1,k_2, l_1, l_2& \in [0,N] 
\end{align*}
where $k$ in $[x_k]_i^j$ (or $[x_k]$) is a dummy variable. $[x_k]$ is a stochastic sequence. The symbols ``$\setminus $" and `` $ ' $ " are used for set subtraction and matrix transposition, respectively. $C_{i,j}$ is a covariance function and $C_i \triangleq C_{i,i}$. $C$ is the covariance matrix of the whole sequence $[x_k]$ (i.e., $C=\text{Cov}(x), x = [x_0',x_1', \cdots, x_N']'$). For a matrix $A$, $A_{[r_1:r_2,c_1:c_2]}$ denotes its submatrix consisting of (block) rows $r_1$ to $r_2$ and (block) columns $c_1$ to $c_2$ of $A$. Also, $0$ may denote a zero scalar, vector, or matrix, as is clear from the context. $F(\cdot | \cdot)$ denotes the conditional cumulative distribution function (CDF). $\mathcal{N}(\mu _k , C_k)$ denotes the Gaussian distribution with mean $\mu _k$ and covariance $C_k$. Also, $\mathcal{N}(x_k;\mu _k , C_k)$ denotes the corresponding Gaussian density with (dummy) variable $x_k$. The abbreviations ZMNG and NG are used for ``zero-mean nonsingular Gaussian" and ``nonsingular Gaussian".

\color{black}

\subsection{Definitions and Notations}\label{Definitions}

Formal (measure-theoretic) definitions of CM (including reciprocal) sequences can be found in \cite{CM_Part_I_Conf}, \cite{Jamison_Reciprocal}, \cite{CM_Part_II_A_Conf}. Here, we present definitions in a simple language.

A sequence $[x_k]$ is $[k_1,k_2]$-$CM_c, c \in \lbrace k_1,k_2 \rbrace$ (i.e., CM over $[k_1,k_2]$) iff conditioned on the state at time $k_1$ (or $k_2$), the sequence is Markov over $[k_1+1,k_2]$ ($[k_1,k_2-1]$). The above definition is equivalent to the following lemma \cite{CM_Part_I_Conf}.

\begin{lemma}\label{CMc_CDF}
$[x_k]$ is $[k_1,k_2]$-$CM_c, c \in \lbrace k_1,k_2 \rbrace$, iff $F(\xi _k|$ $[x_{i}]_{k_1}^{j},x_{c})=F(\xi _k|x_j,x_c)$ for every $j,k \in [k_1,k_2], j<k$, $\forall \xi _k \in \mathbb{R}^d$, where $d$ is the dimension of $x_k$.  

\end{lemma}

The interval $[k_1,k_2]$ of the $[k_1,k_2]$-$CM_c$ sequence is called the \textit{CM interval} of the sequence.  

\begin{remark}\label{R_CMN}
We consider the following notation ($k_1<k_2$)
\begin{align*}
[k_1,k_2]\text{-}CM_c=\left\{ \begin{array}{cc} 
[k_1,k_2]\text{-}CM_F & \text{if  } c=k_1\\ \relax
[k_1,k_2]\text{-}CM_L & \text{if  } c=k_2
\end{array} \right.
\end{align*}
where the subscript ``$F$" or ``$L$" is used because the conditioning is at the \textit{first} or the \textit{last} time of the CM interval. 

\end{remark}

\begin{remark}
When the CM interval of a sequence is the whole time interval, it is dropped: the $[0,N]$-$CM_c$ sequence is called $CM_c$.
\end{remark}

A $CM_0$ sequence is $CM_F$ and a $CM_N$ sequence is $CM_L$. For different values of $k_1$, $k_2$, and $c$, there are different classes of CM sequences. For example, $CM_F$ and $[1,N]$-$CM_L$ are two classes. By a $CM_F \cap [1,N]$-$CM_L$ sequence we mean a sequence being both $CM_F$ and $[1,N]$-$CM_L$. We define that every sequence with a length smaller than 3 (i.e., $\lbrace x_0,x_1 \rbrace$, $\lbrace x_0 \rbrace$, and $\lbrace  \rbrace$) is Markov. Similarly, every sequence is $[k_1,k_2]$-$CM_c$, $|k_2 - k_1| <3$. So, $CM_L$ and $CM_L \cap [k_1,N]$-$CM_F$, $k_1 \in [N-2,N]$ are equivalent.

A sequence is reciprocal iff conditioned on the states at any two times $j$ and $l$, the segment inside the interval $(j,l)$ is independent of the two segments outside $[j,l]$. In other words, inside and outside are independent given the boundaries.

\begin{lemma}\label{CDF}
$[x_k]$ is reciprocal iff $F(\xi _k|[x_{i}]_{0}^{j},[x_i]_l^N)=F(\xi _k|x_j,x_l)$ for every $j,k,l \in [0,N]$ ($j < k < l$), $\forall \xi _k \in \mathbb{R}^d$, where $d$ is the dimension of $x_k$.  

\end{lemma}

\subsection{Preliminaries}\label{Preliminaries}

We review some results required in later sections from \cite{CM_Part_I_Conf}, \cite{CM_Part_II_A_Conf}, \cite{Levy_Dynamic}, \cite{Ackner}.

\begin{theorem}\label{CM_iff_Reciprocal}
$[x_k]$ is reciprocal iff it is $[k_1,N]$-$CM_F$, $\forall k_1 \in [0,N]$, and $CM_L$.

\end{theorem}

\begin{definition}\label{CMc_Matrix}
A symmetric positive definite matrix is called $CM_L$ if it has form $\eqref{CML}$ or $CM_F$ if it has form $\eqref{CMF}$:
\begin{align}
&\left[
\begin{array}{ccccccc}
A_0 & B_0 & 0 & \cdots & 0 & 0 & D_0\\
B_0' & A_1 & B_1 & 0 & \cdots & 0 & D_1\\
0 & B_1' & A_2 & B_2 & \cdots & 0 & D_2\\
\vdots & \vdots & \vdots & \vdots & \vdots & \vdots & \vdots\\
0 & \cdots & 0 & B_{N-3}' & A_{N-2}  & B_{N-2} & D_{N-2}\\
0 & \cdots & 0 & 0 & B_{N-2}' & A_{N-1} & B_{N-1}\\
D_0' & D_1' & D_2' & \cdots & D_{N-2}' & B_{N-1}' & A_N
\end{array}\right]\label{CML}\\
&\left[
\begin{array}{ccccccc}
A_0 & B_0 & D_2 & \cdots & D_{N-2} & D_{N-1} & D_{N}\\
B_0' & A_1 & B_1 & 0 & \cdots & 0 & 0\\
D_2' & B_1' & A_2 & B_2 & \cdots & 0 & 0\\
\vdots & \vdots & \vdots & \vdots & \vdots & \vdots & \vdots\\
D_{N-2}' & \cdots & 0 & B_{N-3}' & A_{N-2}  & B_{N-2} & 0\\
D_{N-1}' & \cdots & 0 & 0 & B_{N-2}' & A_{N-1} & B_{N-1}\\
D_{N}' & 0 & 0 & \cdots & 0 & B_{N-1}' & A_N
\end{array}\right]\label{CMF}
\end{align}

\end{definition}
Here $A_k$, $B_k$, and $D_k$ are matrices in general. We call both $CM_L$ and $CM_F$ matrices $CM_c$. A $CM_c$ matrix is $CM_L$ for $c=N$ and $CM_F$ for $c=0$.

\begin{theorem}\label{CML_Characterization} 
A NG sequence with covariance matrix $C$ is: (i) $CM_c$ iff $C^{-1}$ is $CM_c$, (ii) reciprocal iff $C^{-1}$ is cyclic (block) tri-diagonal (i.e. both $CM_L$ and $CM_F$), (iii) Markov iff $C^{-1}$ is (block) tri-diagonal.

\end{theorem}

\begin{corollary}\label{CML_0k2_Characterization}
A NG sequence with its covariance inverse $C^{-1}=\left[ \begin{array}{cc}
A_{11} & A_{12}\\
A_{21} & A_{22}
\end{array}\right]$ is 

(i) $[0,k_2]$-$CM_c$ ($k_2 \in [1,N-1]$) iff $\Delta _{A_{22}}$ has the $CM_c$ form, where
\begin{align}
\Delta _{A_{22}}&=A_{11}-A_{12}A_{22}^{-1}A_{12}'\label{DA_22}
\end{align}
$A_{11}=A_{[1:k_2+1,1:k_2+1]}$, $A_{22}=A_{[k_2+2:N+1,k_2+2:N+1]}$, and $A_{12}=A_{[1:k_2+1,k_2+2:N+1]}$.

(ii) $[k_1,N]$-$CM_c$ ($k_1 \in [1,N-1]$) iff $\Delta _{A_{11}}$ has the $CM_c$ form, where
\begin{align}
\Delta _{A_{11}}&=A_{22}-A_{12}'A_{11}^{-1}A_{12}\label{DA_11}
\end{align}
$A_{11}=A_{[1:k_1,1:k_1]}$, $A_{22}=A_{[k_1+1:N+1,k_1+1:N+1]}$, and $A_{12}=$ $A_{[1:k_1,k_1+1:N+1]}$. 

\end{corollary}

A positive definite matrix $A$ is called a $[0,k_2]$-$CM_c$ ($[k_1,N]$ -$CM_c$) matrix if $\Delta_{A_{22}}$ ($\Delta_{A_{11}}$) in $\eqref{DA_22}$ ($\eqref{DA_11}$) has the $CM_c$ form.

\begin{theorem}\label{CML_Dynamic_Forward_Proposition}
A ZMNG $[x_k]$ is $CM_c$, $c \in \lbrace 0,N \rbrace$, iff it obeys
\begin{align}
x_k=G_{k,k-1}x_{k-1}+G_{k,c}x_c+e_k, \quad k \in [1,N] \setminus \lbrace c \rbrace
\label{CML_Dynamic_Forward}
\end{align}
where $[e_k]$ ($G_k=\text{Cov}(e_k)$) is a zero-mean white NG sequence, and boundary condition\footnote{Note that $\eqref{CML_Forward_BC1}$ means that for $c=N$ we have $x_0=e_0$ and $x_N=G_{N,0}x_0+e_N$; for $c=0$ we have $x_0=e_0$. Likewise for $\eqref{CML_Forward_BC2}$.}
\begin{align}
&x_c=e_c, \quad x_0=G_{0,c}x_c+e_0 \, \, (\text{for} \, \, c=N) \label{CML_Forward_BC2}
\end{align}
or equivalently\footnote{$e_0$ and $e_N$ in $\eqref{CML_Forward_BC1}$ are not necessarily the same as $e_0$ and $e_N$ in $\eqref{CML_Forward_BC2}$. Just for simplicity we use the same notation.}
\begin{align}
&x_0=e_0, \quad x_c=G_{c,0}x_0+e_c \, \, (\text{for} \,\, c=N)\label{CML_Forward_BC1}
\end{align}

\end{theorem}

\begin{theorem}\label{CML_R_Dynamic_Forward_Proposition}
A ZMNG $[x_k]$ is reciprocal iff it satisfies $\eqref{CML_Dynamic_Forward}$ along with $\eqref{CML_Forward_BC2}$ or $\eqref{CML_Forward_BC1}$, and 
\begin{align}
G_k^{-1}G_{k,c}=G_{k+1,k}'G_{k+1}^{-1}G_{k+1,c}
\label{CML_Condition_Reciprocal}
\end{align}
$\forall k \in [1,N-2]$ for $c=N$, or $\forall k \in [2,N-1]$ for $c=0$. Moreover, for $c=N$, $[x_k]$ is Markov iff in addition to $\eqref{CML_Condition_Reciprocal}$, we have $G_0^{-1}G_{0,N}=G_{1,0}'G_1^{-1}G_{1,N}$ for $\eqref{CML_Forward_BC2}$, or equivalently $G_N^{-1}G_{N,0}=G_{1,N}'G_{1}^{-1}G_{1,0}$ for $\eqref{CML_Forward_BC1}$. Also, for $c=0$, $[x_k]$ is Markov iff in addition to $\eqref{CML_Condition_Reciprocal}$, we have $G_{N,0}=0$.

\end{theorem}

A reciprocal sequence is a special $CM_c$ sequence. Theorem \ref{CML_R_Dynamic_Forward_Proposition} gives a necessary and sufficient condition for a $CM_c$ model to be a model of the ZMNG reciprocal sequence. A model of this sequence was also presented in \cite{Levy_Dynamic}. In other words, the ZMNG reciprocal sequence can be modeled by either what we call the ``reciprocal $CM_c$ model" of Theorem \ref{CML_R_Dynamic_Forward_Proposition} or what we call the ``reciprocal model" of \cite{Levy_Dynamic}.

Similarly, a Markov sequence is a special $CM_c$ sequence. Theorem \ref{CML_R_Dynamic_Forward_Proposition} gives a necessary and sufficient condition for a $CM_c$ model to be a model of the ZMNG Markov sequence. A $CM_c$ model of a Markov sequence is called a ``Markov $CM_c$ model". A different model of the ZMNG Markov sequence is as follows.

\begin{lemma}\label{Markov_Model_Lemma}
A ZMNG $[y_k]$ is Markov iff it obeys
\begin{align}\label{Markov_Model}
y_k=M_{k,k-1}y_{k-1}+e_{k}, \quad k \in [1,N]
\end{align}
where $y_0=e_0$ and $[e_k]$ ($M_k=\text{Cov}(e_k)$) is a zero-mean white NG sequence.

\end{lemma}

\section{Dynamic Models of Reciprocal and Intersections of CM Classes}\label{Section_Model}

\subsection{Reciprocal Sequences}\label{S_M_R}

By Theorem \ref{CML_R_Dynamic_Forward_Proposition}, one can determine whether a $CM_c$ model describes a reciprocal sequence or not. In other words, it gives the required conditions on the parameters of a $CM_c$ model to be a reciprocal $CM_c$ model. However, it does not provide an approach for designing the parameters. Theorem \ref{CML_R_Dynamic_FQ_Proposition} below provides such an approach. First, we need a lemma.

\begin{lemma}\label{Markov_In_CML_Reciprocal_Class}
The set of reciprocal sequences modeled by a reciprocal $CM_L$ model $\eqref{CML_Dynamic_Forward}$ with parameters $(G_{k,k-1},G_{k,N},$ $G_k)$, $k \in [1,N-1]$ includes Markov sequences.

\end{lemma}
\begin{proof}
By Theorem \ref{CML_R_Dynamic_Forward_Proposition}, $\eqref{CML_Dynamic_Forward}$ (for $c=N$) satisfying $\eqref{CML_Condition_Reciprocal}$ with $\eqref{CML_Forward_BC2}$ models a reciprocal sequence. By Theorem \ref{CML_Characterization}, $C^{-1}$ of such a sequence is cyclic (block) tri-diagonal given by $\eqref{CML}$ with $D_1=\cdots=D_{N-2}=0$ and
\begin{align}
D_0=G_{1,0}'G_1^{-1}G_{1,N}-G_0^{-1}G_{0,N}\label{D0_1}
\end{align}  
Now, consider a reciprocal sequence modeled by $\eqref{CML_Dynamic_Forward}$ satisfying $\eqref{CML_Condition_Reciprocal}$ with parameters $(G_{k,k-1},G_{k,N},G_k), k \in [1,N-1]$, and boundary condition $\eqref{CML_Forward_BC2}$ with parameters $G_{0,N}$, $G_0$, and $G_N$, where
\begin{align}
G_{0,N} = G_0G_{1,0}'G_1^{-1}G_{1,N}\label{D_0_0}
\end{align}
meaning that $D_0=0$. This reciprocal sequence is Markov (Theorem \ref{CML_Characterization}). Note that since for every possible value of the parameters of the boundary condition the sequence is nonsingular reciprocal modeled by the same reciprocal $CM_L$ model, choice $\eqref{D_0_0}$ is valid. Thus, there exist Markov sequences belonging to the set of reciprocal sequences modeled by a reciprocal $CM_L$ model $\eqref{CML_Dynamic_Forward}$ with the parameters $(G_{k,k-1},G_{k,N},G_k)$, $k \in [1,N-1]$.  
\end{proof}

\begin{theorem}\label{CML_R_Dynamic_FQ_Proposition}(Markov-induced $CM_L$ model) 
A ZMNG $[x_k]$ is reciprocal iff it can be modeled by a $CM_L$ model $\eqref{CML_Dynamic_Forward}$--$\eqref{CML_Forward_BC2}$ (for $c=N$) induced by a Markov model $\eqref{Markov_Model}$, that is, iff the parameters $(G_{k,k-1},G_{k,N},G_k)$, $k \in [1,N-1]$, of the $CM_L$ model $\eqref{CML_Dynamic_Forward}$--$\eqref{CML_Forward_BC2}$ of $[x_k]$ can be determined by the parameters $(M_{k,k-1},M_k)$, $k \in [1,N]$, of a Markov model $\eqref{Markov_Model}$ as
\begin{align}
G_{k,k-1}&=M_{k,k-1}-G_{k,N}M_{N|k}M_{k,k-1} \label{CML_Choice_1}\\
G_{k,N}&=G_kM_{N|k}'C_{N|k}^{-1} \label{CML_Choice_2}\\
G_k&=(M_k^{-1}+M_{N|k}'C_{N|k}^{-1}M_{N|k})^{-1}\label{CML_Choice_3}
\end{align}
where $M_{N|k}=M_{N,N-1}\cdots M_{k+1,k}$, $M_{N|N}=I$, $C_{N|k}=\sum _{n=k}^{N-1} M_{N|n+1}M_{n+1} M_{N|n+1}'$, $ k \in [1,N-1]$, where $M_{k,k-1}$, $ k \in [1,N]$, are square matrices, and $M_k$, $k \in [1,N]$, are positive definite having the dimension of $x_k$.

\end{theorem}
\begin{proof}
First, we show how $\eqref{CML_Choice_1}$--$\eqref{CML_Choice_3}$ are obtained and prepare the setting for our proof. 

Given the square matrices $M_{k,k-1}, k \in [1,N]$, and the positive definite matrices $M_{k}, k \in [1,N]$, there exists a ZMNG Markov sequence $[y_k]$ (Lemma \ref{Markov_Model_Lemma}):
\begin{align}
y_k=M_{k,k-1}y_{k-1}+e^M_{k}, \quad k \in [1,N], \quad y_0=e^M_0\label{1Order_Markov}
\end{align}
where $[e^M_k]$ is a zero-mean white NG sequence with covariances $M_k, k \in [0,N]$.  

Since every Markov sequence is $CM_L$, we can obtain a $CM_L$ model of $[y_k]$ as
\begin{align}
y_k=G_{k,k-1}y_{k-1}+G_{k,N}y_N+e^y_k, \quad k \in [1,N-1]\label{CML_Dynamic_for_Markov}
\end{align}
where $[e^y_k]$ is a zero-mean white NG sequence with covariances $G_k, k \in [1,N-1], G^y_0, G^y_N$, and boundary condition
\begin{align}
y_N&=e^{y}_N, \quad y_0=G^{y}_{0,N}y_N+e^{y}_0\label{CML_R_FQ_BC2}
\end{align}

Parameters of $\eqref{CML_Dynamic_for_Markov}$ can be obtained as follows. By $\eqref{1Order_Markov}$, we have $p(y_k|y_{k-1})=\mathcal{N}(y_k;M_{k,k-1}y_{k-1},M_{k})$. Since $[y_k]$ is Markov, we have, for $\forall k \in [1,N-1]$,
\begin{align}
p(y_k|y_{k-1},y_N)&=\frac{p(y_k|y_{k-1})p(y_N|y_k,y_{k-1})}{p(y_N|y_{k-1})}\nonumber\\
&=\frac{p(y_k|y_{k-1})p(y_N|y_k)}{p(y_N|y_{k-1})}\label{CML_Reciprocal_Transition}\\
&=\mathcal{N}(y_k;G_{k,k-1}y_{k-1}+G_{k,N}y_N,G_k)\nonumber
\end{align}
and it turns out that $G_{k,k-1}$, $G_{k,N}$, and $G_k$ are given by $\eqref{CML_Choice_1}$--$\eqref{CML_Choice_3}$ \cite{Gaussian_2}, where we have $p(y_k|y_{k-1})=\mathcal{N}(y_k;M_{k,k-1}y_{k-1},M_k)$.

Now, we construct a sequence $[x_k]$ modeled by the same model $\eqref{CML_Dynamic_for_Markov}$ as
\begin{align}
x_k=G_{k,k-1}x_{k-1}+G_{k,N}x_N+e_k, \quad k \in [1,N-1]\label{CML_Dynamic_for_Markov_x}
\end{align}
where $[e_k]$ is a zero-mean white NG sequence with covariances $G_k, k \in [0,N]$, and boundary condition
\begin{align}
x_N&=e_N, \quad x_0=G_{0,N}x_N+e_0\label{CML_R_FQ_BC2_x}
\end{align}
but with different parameters of the boundary condition (i.e., $(G_N,G_{0,N},G_0) \neq (G^y_N,G^y_{0,N},G^y_0)$). By Theorem \ref{CML_Dynamic_Forward_Proposition}, $[x_k]$ is a ZMNG $CM_L$ sequence. Note that parameters of $\eqref{CML_Dynamic_for_Markov}$ and $\eqref{CML_Dynamic_for_Markov_x}$ are the same ($G_{k,k-1}, G_{k,N},$ $ G_k, k \in [1,N-1]$), but parameters of $\eqref{CML_R_FQ_BC2}$ ($G^y_{0,N}, G^y_0,G^y_N$) and $\eqref{CML_R_FQ_BC2_x}$ ($G_{0,N}, G_0,G_N$) are different.

Sufficiency: we prove sufficiency; that is, a $CM_L$ model with the parameters $\eqref{CML_Choice_1}$--$\eqref{CML_Choice_3}$ is a reciprocal $CM_L$ model. It suffices to show that the parameters $\eqref{CML_Choice_1}$--$\eqref{CML_Choice_3}$ satisfy $\eqref{CML_Condition_Reciprocal}$ and consequently $[x_k]$ is reciprocal. Substituting $\eqref{CML_Choice_1}$--$\eqref{CML_Choice_3}$ in $\eqref{CML_Condition_Reciprocal}$, for the right hand side of $\eqref{CML_Condition_Reciprocal}$, we have
\begin{align*}
&G_{k+1,k}'G_{k+1}^{-1}G_{k+1,N}=M_{N|k}'C_{N|k+1}^{-1}-M_{N|k}'C_{N|k+1}^{-1}M_{N|k+1}\\
&\cdot(M_{k+1}^{-1}+M_{N|k+1}'C_{N|k+1}^{-1}M_{N|k+1})^{-1}M_{N|k+1}'C_{N|k+1}^{-1}
\end{align*}
and for the left hand side of $\eqref{CML_Condition_Reciprocal}$, we have $G_k^{-1}G_{k,N}=M_{N|k}'C_{N|k}^{-1}=M_{N|k}'(C_{N|k+1}+M_{N|k+1}M_{k+1}M_{N|k+1}')^{-1}$, where from the matrix inversion lemma it follows that $\eqref{CML_Condition_Reciprocal}$ holds. Therefore, $[x_k]$ is reciprocal. So, equations $\eqref{CML_Dynamic_Forward}$--$\eqref{CML_Forward_BC2}$ with $\eqref{CML_Choice_1}$--$\eqref{CML_Choice_3}$ model a ZMNG reciprocal sequence.

Necessity: Let $[x_k]$ be ZMNG reciprocal. By Theorem \ref{CML_R_Dynamic_Forward_Proposition} $[x_k]$ obeys $\eqref{CML_Dynamic_Forward}$--$\eqref{CML_Forward_BC2}$ with $\eqref{CML_Condition_Reciprocal}$. By Lemma \ref{Markov_In_CML_Reciprocal_Class}, the set of reciprocal sequences modeled by a reciprocal $CM_L$ model contains Markov and non-Markov sequences (depending on the parameters of the boundary condition). So, a sequence modeled by a reciprocal $CM_L$ model and a boundary condition determined as in the proof of Lemma \ref{Markov_In_CML_Reciprocal_Class} (i.e., satisfying $\eqref{D_0_0}$) is actually a Markov sequence whose $C^{-1}$ is (block) tri-diagonal (i.e., $\eqref{CML}$ with $D_0=\cdots= D_{N-2}=0$). Given this $C^{-1}$, we can obtain parameters of Markov model $\eqref{1Order_Markov}$ ($M_{k,k-1}, k \in [1,N]$, $M_{k}, k \in [0,N]$) of a Markov sequence with the given $C^{-1}$ as follows. $C^{-1}$ of a Markov sequence can be calculated in terms of parameters of a Markov $CM_L$ model or those of a Markov model. Equating these two formulations of $C^{-1}$, parameters of the Markov model are obtained in terms of those of the Markov $CM_L$ model. Thus, for $k=N-2, N-3, \ldots , 0$,
\begin{align}
M_{N}^{-1}&=A_N\label{PM_1}\\
M_{N,N-1}&=-M_{N}B_{N-1}'\label{PM_2}\\
M_{k+1}^{-1}&=A_{k+1}-M_{k+2,k+1}'M_{k+2}^{-1}M_{k+2,k+1}\label{PM_3}\\
M_{k+1,k}&=-M_{k+1}B_k' \label{PM_4}\\
M_0^{-1}&=A_0-M_{1,0}'M_{1}^{-1}M_{1,0}\label{PM_5}
\end{align}
where
\begin{align}
A_0&=G_0^{-1}+G_{1,0}'G_1^{-1}G_{1,0}\label{CML1}\\
A_k&=G_k^{-1}+G_{k+1,k}'G_{k+1}^{-1}G_{k+1,k}, k \in [1,N-2]\label{CML2}\\
A_{N-1}&=G_{N-1}^{-1}\label{CML3}\\
A_{N}&=G_N^{-1}+\sum _{k=0}^{N-1} G_{k,N}'G_k^{-1}G_{k,N}\label{CML4}\\
B_k&=-G_{k+1,k}'G_{k+1}^{-1}, k \in [0,N-2]\label{CML5}\\
B_{N-1}&=-G_{N-1}^{-1}G_{N-1,N}\label{CML6}
\end{align}

Following $\eqref{CML_Reciprocal_Transition}$ to get a reciprocal $CM_L$ model from this Markov model, we have $\eqref{CML_Choice_1}$--$\eqref{CML_Choice_3}$.

What remains to be proven is that the parameters of the model obtained by $\eqref{CML_Choice_1}$--$\eqref{CML_Choice_3}$ are the same as those of the $CM_L$ model calculated directly based on the covariance function of $[x_k]$. By Theorem \ref{CML_Dynamic_Forward_Proposition}, the model constructed from $\eqref{CML_Choice_1}$--$\eqref{CML_Choice_3}$ is a valid $CM_L$ model. In addition, given a $CM_L$ matrix (a positive definite cyclic (block) tri-diagonal matrix is a special $CM_L$ matrix) as the $C^{-1}$ of a sequence, the set of parameters of the $CM_L$ model and boundary condition of the sequence is unique (it can be seen by the almost sure uniqueness of a conditional expectation \cite{CM_Part_I_Conf}). Thus, the parameters $\eqref{CML_Choice_1}$--$\eqref{CML_Choice_3}$ must be the same as those obtained directly from the covariance function of $[x_k]$. Thus, a ZMNG reciprocal sequence $[x_k]$ obeys $\eqref{CML_Dynamic_Forward}$--$\eqref{CML_Forward_BC2}$ with $\eqref{CML_Choice_1}$--$\eqref{CML_Choice_3}$.
\end{proof}

%===================================================
% I edited the formula : added M_{N|k}M_{k}
%===================================================

Note that by matrix inversion lemma, $\eqref{CML_Choice_3}$ is equivalent to $G_k=M_{k} - M_{k}M_{N|k}'(C_{N|k} + M_{N|k}M_{k}M_{N|k}')^{-1}$ $\cdot M_{N|k}M_{k}$.

%===================================================
% I edited the formula : added M_{N|k}M_{k}
%===================================================

Note that Theorem \ref{CML_R_Dynamic_FQ_Proposition} holds true for every combination of the parameters (i.e., square matrices $M_{k,k-1}$ and positive definite matrices $M_k, k \in [1,N]$). By $\eqref{CML_Choice_1}$--$\eqref{CML_Choice_3}$, parameters of every reciprocal $CM_L$ model are obtained from $M_{k,k-1}, M_k, k \in [1,N]$, which are parameters of a Markov model $\eqref{Markov_Model}$. This is particularly useful for parameter design of a reciprocal $CM_L$ model. We explain it for the problem of motion trajectory modeling with destination information as follows. Such trajectories can be modeled by combining two key assumptions: (i) the object motion follows a Markov model $\eqref{Markov_Model}$ (e.g., a nearly constant velocity model) without considering the destination information, and (ii) the joint origin and destination density is known (which can be different from that of the Markov model in (i)). In reality, if the joint density is not known, an approximate density can be used. Now, (by (i)) let $[y_k]$ be Markov modeled by $\eqref{1Order_Markov}$ (e.g., a nearly constant velocity model without considering the destination information) with parameters $M_{k,k-1}, k \in [1,N], M_{k}, k \in [1,N]$. $[y_k]$ can be also modeled by a $CM_L$ model $\eqref{CML_Dynamic_for_Markov}$--$\eqref{CML_R_FQ_BC2}$. By the Markov property, parameters of $\eqref{CML_Dynamic_for_Markov}$ are obtained as $\eqref{CML_Choice_1}$--$\eqref{CML_Choice_3}$ based on  $\eqref{CML_Reciprocal_Transition}$. Next, we construct $[x_k]$ modeled by $\eqref{CML_Dynamic_for_Markov_x}$--$\eqref{CML_R_FQ_BC2_x}$. By Theorem \ref{CML_Dynamic_Forward_Proposition}, $[x_k]$ is a $CM_L$ sequence. Since parameters of $\eqref{CML_R_FQ_BC2_x}$ are arbitrary, $[x_k]$ can have any joint density of $x_0$ and $x_N$. So, $[y_k]$ and $[x_k]$ have the same $CM_L$ model ($\eqref{CML_Dynamic_for_Markov}$ and $\eqref{CML_Dynamic_for_Markov_x}$) (i.e., the same transition $\eqref{CML_Reciprocal_Transition}$), but $[x_k]$ can have any joint distribution of the states at the endpoints. In other words, $[x_k]$ can model any origin and destination. Therefore, combining the two assumptions (i) and (ii) above naturally leads to a $CM_L$ sequence $[x_k]$ whose $CM_L$ model is the same as that of $[y_k]$ while the former can model any origin and destination. Thus, model $\eqref{CML_Dynamic_for_Markov_x}$ with $\eqref{CML_Choice_1}$--$\eqref{CML_Choice_3}$ is the desired model for destination-directed trajectory modeling based on (i) and (ii) above.

Markov sequences modeled by the same reciprocal model of \cite{Levy_Dynamic} were studied in \cite{Levy_Class}. This is an important topic in the theory of reciprocal processes \cite{Jamison_Reciprocal}. In the following, Markov sequences modeled by the same $CM_L$ model $\eqref{CML_Dynamic_Forward}$ are studied and determined. Following the notion of a reciprocal transition density derived from a Markov transition density \cite{Jamison_Reciprocal}, a $CM_L$ model \textit{induced} by a Markov model is defined as follows. A Markov sequence can be modeled by either a Markov model $\eqref{Markov_Model}$ or a $CM_L$ model $\eqref{CML_Dynamic_Forward}$. Such a $CM_L$ model is called the $CM_L$ model \textit{induced} by the Markov model since parameters of the former can be obtained from those of the latter (see $\eqref{CML_Reciprocal_Transition}$ or $\eqref{PM_1}$--$\eqref{CML6}$). Definition \ref{CML_Derived} is for the Gaussian case.

\begin{definition}\label{CML_Derived}
Consider a Markov model $\eqref{Markov_Model}$ with parameters $M_{k,k-1}, k \in [1,N], M_{k}, k \in [1,N]$. The $CM_L$ model $\eqref{CML_Dynamic_Forward}$ with parameters $(G_{k,k-1},G_{k,N},G_k)$, $k \in [1,N-1]$, given by $\eqref{CML_Choice_1}$--$\eqref{CML_Choice_3}$ is called the \textit{Markov-induced $CM_L$ model}. 

\end{definition}

\begin{corollary}\label{Reciprocal_Derived_Markov}
A $CM_L$ model $\eqref{CML_Dynamic_Forward}$ is for a reciprocal sequence iff it can be so induced by a Markov model $\eqref{Markov_Model}$. 

\end{corollary}
\begin{proof}
See our proof of Theorem \ref{CML_R_Dynamic_FQ_Proposition}.
\end{proof} 

By the proof of Theorem \ref{CML_R_Dynamic_FQ_Proposition}, given a reciprocal $CM_L$ model $\eqref{CML_Dynamic_Forward}$ (satisfying $\eqref{CML_Condition_Reciprocal}$), we can choose a boundary condition satisfying $\eqref{D_0_0}$ and then obtain a Markov model $\eqref{Markov_Model}$ for a Markov sequence that obeys the given reciprocal $CM_L$ model (see $\eqref{PM_1}$--$\eqref{CML6}$). Since parameters of the boundary condition (i.e., $G_{0,N}$, $G_0$, and $G_N$) satisfying $\eqref{D_0_0}$ can take many values, there are many such Markov models and their parameters are given by $\eqref{PM_1}$--$\eqref{PM_5}$.

The idea of obtaining a reciprocal evolution law from a Markov evolution law was used in \cite{Schrodinger_1}, \cite{Jamison_Reciprocal}, and later for finite-state reciprocal sequences in \cite{Fanas1}, \cite{White_2}. Our contributions are different. First, our reciprocal $CM_L$ model above is from the CM viewpoint. Second, Theorem \ref{CML_R_Dynamic_FQ_Proposition} not only induces a reciprocal $CM_L$ model by a Markov model, but also shows that \textit{every} reciprocal $CM_L$ model can be induced by a Markov model (by necessity and sufficiency of Theorem \ref{CML_R_Dynamic_FQ_Proposition}). This is important for application of reciprocal sequences (i.e., parameter design of a reciprocal $CM_L$ model) because one usually has a much better intuitive understanding of Markov models (see the above explanation for trajectory modeling with reciprocal sequences). Third, our proof of Theorem \ref{CML_R_Dynamic_FQ_Proposition} is constructive and shows how a given reciprocal $CM_L$ model can be induced by a Markov model. Fourth, our constructive proof of Theorem \ref{CML_R_Dynamic_FQ_Proposition} gives all possible Markov models by which a given reciprocal $CM_L$ model can be induced. Note that only one $CM_L$ model can be induced by a given Markov model (it can be verified by $\eqref{PM_1}$--$\eqref{CML6}$). However, a given reciprocal $CM_L$ model can be induced by many different Markov models. This is because $\eqref{D_0_0}$ holds for many different choices of parameters of the boundary condition (i.e., $G_{0,N}$, $G_0$, and $G_N$) each of which leads to a Markov model with parameters given by $\eqref{PM_1}$--$\eqref{PM_5}$ (see the proof of necessity of Theorem \ref{CML_R_Dynamic_FQ_Proposition}). By Theorem \ref{CML_R_Dynamic_FQ_Proposition}, one can see that the transition density of the bridging distribution used in \cite{Simon}--\cite{Simon2} is a reciprocal transition density.

\subsection{Intersections of CM Classes}\label{S_M_I}

In some applications sequences with more than one CM property (i.e., belonging to more than one CM class) are desired. An example is trajectories with a waypoint and a destination information. Assume we know not only the destination density (at time $N$) but also the state density at time $k_2(<N)$ (i.e., waypoint information). First, consider only the waypoint information at $k_2$ (without destination information). In other words, we know the state density at $k_2$ but not after. With a CM evolution law between $0$ and $k_2$, such trajectories can be modeled as a $[0,k_2]$-$CM_L$ sequence. Now, consider only the destination information (density) without waypoint information. Such trajectories can be modeled as a $CM_L$ sequence. Then, trajectories with a waypoint and a destination information can be modeled as a sequence being both $[0,k_2]$-$CM_L$ and $CM_L$, denoted as $CM_L \cap [0,k_2]$-$CM_L$. In other words, the sequence has both the $CM_L$ property and the $[0,k_2]$-$CM_L$ property. Studying the evolution of other sequences belonging to more than one CM class, for example $CM_L \cap [k_1,N]$-$CM_F$, is also useful for studying reciprocal sequences. The NG reciprocal sequence is equivalent to $CM_L \cap CM_F$ \cite{CM_Part_II_A_Conf}. Proposition \ref{CML_k1N_CMF_Dynamic} below presents a dynamic model of $CM_L \cap [k_1,N]$-$CM_F$ sequences, based on which one can see a full spectrum of models from a $CM_L$ sequence to a reciprocal sequence.

\begin{proposition}\label{CML_k1N_CMF_Dynamic}
A ZMNG $[x_k]$ is $CM_L \cap [k_1,N]$-$CM_F$ iff it obeys $\eqref{CML_Dynamic_Forward}$--$\eqref{CML_Forward_BC2}$ with ($\forall k \in [k_1+1,N-2]$)
\begin{align}
G_k^{-1}G_{k,N}=G_{k+1,k}'G_{k+1}^{-1}G_{k+1,N}
\label{CML_Condition_k1N_CMF}
\end{align}

\end{proposition}
\begin{proof}
A ZMNG $CM_L$ sequence has a $CM_L$ model $\eqref{CML_Dynamic_Forward}$--$\eqref{CML_Forward_BC2}$ (Theorem \ref{CML_Dynamic_Forward_Proposition}). Also, a NG sequence is $[k_1,N]$-$CM_F$ iff its $C^{-1}$ has the $[k_1,N]$-$CM_F$ form (Corollary \ref{CML_0k2_Characterization}). Then, a sequence is $CM_L \cap [k_1,N]$-$CM_F$ iff it obeys $\eqref{CML_Dynamic_Forward}$--$\eqref{CML_Forward_BC2}$, where $C^{-1}$ of the sequence has the $[k_1,N]$-$CM_F$ form, which is equivalent to $\eqref{CML_Condition_k1N_CMF}$.
\end{proof}

Proposition \ref{CML_k1N_CMF_Dynamic} shows how models change from a $CM_L$ model to a reciprocal $CM_L$ model for $k_1=0$ (compare $\eqref{CML_Condition_k1N_CMF}$ and $\eqref{CML_Condition_Reciprocal}$ (for $c=N$)). Note that $CM_L$ and $CM_L \cap [k_1,N]$-$CM_F$, $k_1 \in [N-2,N]$ are equivalent (Subsection \ref{Definitions}).

Following the idea of the proof of Proposition \ref{CML_k1N_CMF_Dynamic}, we can obtain models for intersections of different CM classes, for example $CM_c \cap [k_1,k_2]$-$CM_c \cap [m_1,m_2]$-$CM_c$ sequences. However, the above approach does not lead to simple results in some cases, e.g., $CM_L \cap [0,k_2]$-$CM_L$ sequences. A different way of obtaining a model for $CM_L \cap [0,k_2]$-$CM_L$ sequences is presented next.

\begin{proposition}\label{CML_0k2_CML_Dynamic_2nd} 
A ZMNG $[x_k]$ is $CM_L \cap [0,k_2]$-$CM_L$ iff
\begin{align}
x_k&=G_{k,k-1}x_{k-1}+G_{k,k_2}x_{k_2}+e_k, k \in [1,k_2-1] \label{ee1}\\
x_{k_2}&=e_{k_2}, \quad x_0=G_{0,k_2}x_{k_2}+e_{0}\label{ee2}\\
x_N&=\sum _{i=0}^{k_2} G_{N,i}x_{i}+e_N\label{ee3}\\
x_k&=G_{k,k-1}x_{k-1}+G_{k,N}x_{N}+e_k, k \in [k_2+1,N-1] \label{ee4}
\end{align}
where $[e_k]$ ($G_k=\text{Cov}(e_k)$) is a zero-mean white NG sequence, 
\begin{align}
&G_{N,j}'G_{N}^{-1}G_{N,i}=0\label{ee6}\\
&G_{l}^{-1}G_{l,k_2}=
G_{l+1,l}'G_{l+1}^{-1}G_{l+1,k_2}+
G_{N,l}'G_N^{-1}G_{N,k_2}\label{ee7}
\end{align}
$j=0,\ldots,k_2-3$, $i=j+2,\ldots,k_2-1$, and $l=0,\ldots,k_2-2$.

\end{proposition}
\begin{proof}
Necessity: Let $[x_k]$ be a ZMNG $CM_L \cap [0,k_2]$-$CM_L$ sequence. Let $p(\cdot )$ and $p(\cdot| \cdot)$ be its density and conditional density, respectively. Then,
\begin{align}
& x_{k_2} \sim p(x_{k_2})\label{e1}\\
& x_0 \sim p(x_0|x_{k_2})\label{e2}
\end{align}
Since $[x_k]$ is $CM_L \cap [0,k_2]$-$CM_L$, it is $[0,k_2]$-$CM_L$. Thus, for $k \in [1,k_2-1]$, 
\begin{align}
x_k \sim p(x_k|x_0,\ldots ,x_{k-1},x_{k_2})=p(x_k|x_{k-1},x_{k_2})\label{e3}
\end{align}
Also, since $[x_k]$ is $CM_L$, for $k \in [k_2+1,N]$,
\begin{align}
&x_N \sim p(x_N|x_0,\ldots ,x_{k_2})\label{e4}\\
&x_k \sim p(x_k|x_0,\ldots ,x_{k-1},x_N)=p(x_k|x_{k-1},x_N)\label{e5}
\end{align}

According to $\eqref{e1}$--$\eqref{e2}$, we have $x_{k_2}=e_{k_2}$ and $x_0=G_{0,k_2}x_{k_2}+e_0$, where $e_0$ and $e_{k_2}$ are uncorrelated ZMNG with nonsingular covariances $G_0$ and $G_{k_2}$, $G_{0,k_2}=C_{0,k_2}C_{k_2}^{-1}$, $G_{k_2}=C_{k_2}$, $G_0=C_{0}-C_{0,k_2}C_{k_2}^{-1}C_{0,k_2}'$, and $C_{l_1,l_2}$ is the covariance function of $[x_k]$. For $k \in [1,k_2-1]$, by $\eqref{e3}$ we have $x_k=G_{k,k-1}x_{k-1}+G_{k,k_2}x_{k_2}+e_k$, $G_k=\text{Cov}(e_k)$ (see \cite{CM_Part_I_Conf}), $[G_{k,k-1} , G_{k,k_2}]= [C_{k,k-1} , C_{k,k_2}] \left[
\begin{array}{cc}
C_{k-1} & C_{k-1,k_2}\\
C_{k_2,k-1} & C_{k_2}
\end{array} \right]^{-1}$, and $G_k=C_{k}-[C_{k,k-1}, C_{k,k_2}]\left[
\begin{array}{cc}
C_{k-1} & C_{k-1,k_2}\\
C_{k_2,k-1} & C_{k_2}
\end{array} \right]^{-1}
[C_{k,k-1} , C_{k,k_2}]'$.

For $k \in [k_2+1,N]$, by $\eqref{e4}$ we have $x_N=\sum _{i=0}^{k_2} G_{N,i}x_i +$ $ e_N$, $G_N=\text{Cov}(e_N)$, $[G_{N,0} , G_{N,1} , \ldots , G_{N,k_2}] =
C_{[N+1:N+1,1:k_2+1]} (C_{[1:k_2+1,1:k_2+1]})^{-1}$, and $G_{N}=C_{N} -  C_{[N}$ $_{+1:N+1,1:k_2+1]} (C_{[1:k_2+1,1:k_2+1]})^{-1} C_{[N+1:N+1,1:k_2+1]}'$. Here, $C_{[r_1:r_2,c_1:c_2]}$ denotes the submatrix of the covariance matrix $C$ of $[x_k]$ including the block rows $r_1$ to $r_2$ and the block columns $c_1$ to $c_2$. \footnote{Note that $C$ is an $(N+1) \times (N+1)$ matrix for a scalar sequence.}

By $\eqref{e5}$ we have $x_k=G_{k,k-1}x_{k-1}+G_{k,N}x_{N}+e_k,$ $k \in [k_2+1,N-1]$, $G_k=\text{Cov}(e_k)$, $[G_{k,k-1} , G_{k,N}]=[C_{k,k-1} , C_{k,N}] \left[
\begin{array}{cc}
C_{k-1} & C_{k-1,N}\\
C_{N,k-1} & C_{N}
\end{array} \right]^{-1}$, and $G_k=C_{k}-[C_{k,k-1} , C_{k,N}] \left[
\begin{array}{cc}
C_{k-1} & C_{k-1,N}\\
C_{N,k-1} & C_{N}
\end{array} \right]^{-1}
[C_{k,k-1} , C_{k,N}]'$. In the above, $[e_k]$ is a zero-mean white NG sequence with covariances $G_k$. 

Now we show that $\eqref{ee6}$--$\eqref{ee7}$ hold. We construct $C^{-1}$ of the whole sequence $[x_k]$ and obtain $\eqref{ee6}$--$\eqref{ee7}$ from the fact that $C^{-1}$ is both $CM_L$ and $[0,k_2]$-$CM_L$. $[x_k]_0^{k_2}$ obeys $\eqref{ee1}$--$\eqref{ee2}$. So, by Theorem \ref{CML_Dynamic_Forward_Proposition}, $[x_k]_0^{k_2}$ is $CM_L$. Then, by Theorem \ref{CML_Characterization}, $(C_{[1:k_2+1,1:k_2+1]})^{-1}$ is $CM_L$ for every value of parameters of $\eqref{ee1}$--$\eqref{ee2}$ (i.e., $C^{-1}$ is $[0,k_2]$-$CM_L$). $C^{-1}$ of $[x_k]$ is calculated by stacking $\eqref{ee1}$--$\eqref{ee4}$ as follows. We have
\begin{align}
\mathcal{G}x&=e\label{BigLx=e}
\end{align}
where $x \triangleq [x_0' , x_1' , \ldots , x_{N}']'$, $e\triangleq [e_0' , e_1' , \ldots  , e_{N}']'$, $\mathcal{G}=\left[ \begin{array}{cc}
\mathcal{G}_{11} & 0\\
\mathcal{G}_{21} & \mathcal{G}_{22}
\end{array}\right]$,  

\begin{align*}
\mathcal{G}_{21}&=\left[ \begin{array}{cccc}
0 & \cdots & 0 & -G_{k_2+1,k_2}\\
0 & \cdots & 0 & 0\\
\vdots & \vdots & \vdots & \vdots \\
-G_{N,0} & \cdots & -G_{N,k_2-1} & -G_{N,k_2}
\end{array}\right]\nonumber
\end{align*}
$\mathcal{G}_{11}$ is
\begin{align*}
\left[ \begin{array}{cccccc}
I & 0 & 0 &  \cdots & 0 & -G_{0,k_2}\\
-G_{1,0} & I & 0 &  \cdots & 0 & -G_{1,k_2}\\
0 & -G_{2,0} & I & 0 & \cdots & -G_{2,k_2}\\
\vdots & \vdots & \vdots & \vdots & \vdots & \vdots \\
0 & 0 & \cdots & -G_{k_2-1,k_2-2} & I & -G_{k_2-1,k_2}\\
0 & 0 & 0 &  \cdots & 0 & I
\end{array}\right]
\end{align*}
and $\mathcal{G}_{22}$ is
\begin{align*}
\left[ \begin{array}{ccccc}
I & 0 & \cdots & 0 & -G_{k_2+1,N}\\
-G_{k_2+2,k_2+1} & I & 0 & \cdots & -G_{k_2+2,N}\\
\vdots & \vdots & \vdots & \vdots & \vdots \\
0 & \cdots & -G_{N-1,N-2} & I & -G_{N-1,N}\\
0 & \cdots & 0 & 0 & I
\end{array}\right]
\end{align*}
Then,
\begin{align}
C^{-1}=\mathcal{G}'G^{-1}\mathcal{G}\label{CML_Cinv}
\end{align}
where $G=\text{diag}(G_0,G_1,\ldots,G_{N})$. Since $[x_k]$ is $CM_L$, $C^{-1}$ has the $CM_L$ form, which is equivalent to $\eqref{ee6}$--$\eqref{ee7}$.

Sufficiency: We need to show that a sequence modeled by $\eqref{ee1}$-$\eqref{ee7}$ is $CM_L \cap [0,k_2]$-$CM_L$, that is, its $C^{-1}$ has both $CM_L$ and $[0,k_2]$-$CM_L$ forms. Since $[x_k]_0^{k_2}$ obeys $\eqref{ee1}$--$\eqref{ee2}$, $(C_{[1:k_2+1,1:k_2+1]})^{-1}$ has the $CM_L$ form for every choice of parameters of $\eqref{ee1}$--$\eqref{ee2}$ (Theorem \ref{CML_Dynamic_Forward_Proposition} and Theorem \ref{CML_Characterization}). So, $[x_k]$ governed by $\eqref{ee1}$-$\eqref{ee7}$ is $[0,k_2]$-$CM_L$. Also, $C^{-1}$ can be calculated by $\eqref{CML_Cinv}$. It can be seen that $\eqref{ee6}$--$\eqref{ee7}$ is equivalent to $C^{-1}$ having the $CM_L$ form. Thus, a sequence modeled by $\eqref{ee1}$-$\eqref{ee7}$ is $CM_L \cap [0,k_2]$-$CM_L$. The Gaussianity of $[x_k]$ follows clearly from linearity of $\eqref{ee1}$-$\eqref{ee4}$. Also, $[x_k]$ is nonsingular due to $\eqref{CML_Cinv}$, the nonsingularity of $\mathcal{G}$, and the positive definiteness of $G$.   
\end{proof}

\section{Representations of CM and Reciprocal Sequences}\label{Section_Rep}

A representation of NG continuous-time CM processes in terms of a Wiener process and an uncorrelated NG vector was presented in \cite{ABRAHAM}. Inspired by this representation, we show that a NG $CM_c$ sequence can be represented by a sum of a NG Markov sequence and an uncorrelated NG vector. We also show how to use a NG Markov sequence and an uncorrelated NG vector to construct a NG $CM_c$ sequence. 

\begin{proposition}\label{CML_Markov_z_Proposition}
A ZMNG $[x_k]$ is $CM_c$ iff it can be represented as 
\begin{align}\label{CML_Markov_z}
x_k=y_k+ \Gamma _k x_c, \quad k \in [0,N] \setminus \lbrace c \rbrace
\end{align}
where $[y_k] \setminus \lbrace y_c \rbrace $ \footnote{For $c=N$, $[y_k] \setminus \lbrace y_c \rbrace =[y_k]_0^{N-1}$, and for $c=0$, $[y_k] \setminus \lbrace y_c \rbrace =[y_k]_1^N$.} is a ZMNG Markov sequence, $x_c$ is a ZMNG vector uncorrelated with $[y_k] \setminus \lbrace y_c \rbrace$, and $\Gamma_k$ are some matrices. 

\end{proposition}
\begin{proof}
Let $c=N$. Necessity: We show that any ZMNG $CM_L$ $[x_k]$ can be represented as $\eqref{CML_Markov_z}$. Any such $[x_k]$ obeys
\begin{align}
x_k&=G_{k,k-1}x_{k-1}+G_{k,N}x_N+e_k, \quad k \in [1,N-1]\label{CML_Dynamic_Forward_1}\\
x_0&=G_{0,N}x_N+e_0 \label{CML_Forward_BC2_11}\\
x_N&=e_N \label{CML_Forward_BC2_22}
\end{align} 
where $[e_k]$ ($G_k=\text{Cov}(e_k)$) is zero-mean white NG.

According to $\eqref{CML_Forward_BC2_11}$, we consider $y_0=e_0$ and $\Gamma _0=G_{0,N}$. So, $x_0=y_0+\Gamma _0x_N$. For $k \in [1,N-1]$, we have
\begin{align*}
x_k&=G_{k,k-1}x_{k-1}+G_{k,N}x_N+e_k\\
&=G_{k,k-1}(y_{k-1}+\Gamma _{k-1}x_N)+G_{k,N}x_N+e_k\\
&=G_{k,k-1}y_{k-1}+e_k+(G_{k,k-1}\Gamma_{k-1}+G_{k,N})x_N
\end{align*}
By induction, $[x_k]$ can be represented as $x_k=y_k+\Gamma _kx_N$, $k \in [0,N-1]$, where for $k \in [1,N-1]$, $y_k=U_{k,k-1}y_{k-1}+e_{k}$, $U_{k,k-1}=G_{k,k-1}$, $\Gamma _k=G_{k,k-1}\Gamma _{k-1}+G_{k,N}$, $y_0=e_0$, $\Gamma _0=G_{0,N}$, and $x_N$ is uncorrelated with the Markov sequence $[y_k]_0^{N-1}$, because $x_N$ is uncorrelated with $[e_k]_0^{N-1}$. 

What remains is to show the nonsingularity of $[y_k]_0^{N-1}$ and the random vector $x_N$. Since the sequence $[x_k]$ is nonsingular, $x_N$ is nonsingular. Also, we have $y_0=e_0$. In addition, the covariances $G_k, k \in [0,N]$, are nonsingular. Thus, $U_k=\text{Cov}(e_k)$, $k \in [0,N-1]$, are all nonsingular. Similar to $\eqref{CML_Cinv}$, we have $C^\mathsf{y}=\text{Cov}(\mathsf{y})=W^{-1}UW'^{-1}$,
where $\mathsf{y}=[y_0',y_1',\ldots , y_{N-1}']'$, $U=\text{diag}(U_0,U_1,\ldots,U_{N-1})$ and $W$ is a nonsingular matrix. Therefore, $[y_k]_0^{N-1}$ is nonsingular because $U$ and $W$ are nonsingular.

Sufficiency: We show that for any ZMNG Markov sequence $[y_k]_0^{N-1}$ uncorrelated with any ZMNG vector $x_N$, $[x_k]$ constructed as $x_k=y_k+\Gamma _kx_N, k \in [0,N-1]$ is a ZMNG $CM_L$ sequence, where $\Gamma_k$ are some matrices. Therefore, it suffices to show that $[x_k]$ obeys $\eqref{CML_Dynamic_Forward}$--$\eqref{CML_Forward_BC2}$. Since $[y_k]_0^{N-1}$ is a ZMNG Markov sequence, it obeys (Lemma \ref{Markov_Model_Lemma}) $y_k=U_{k,k-1}y_{k-1}+e_{k}$, $ k \in [1,N-1]$, $y_0=e_0$, where $[e_{k}]_0^{N-1}$ is a zero-mean white NG sequence with covariances $U_{k}$. 

We have $x_0=y_0+\Gamma _0x_N$. So, consider $G_{0,N}=\Gamma _0$. Then, for $k \in [1,N-1]$, we have
\begin{align}
x_k&=y_k+\Gamma _kx_N=U_{k,k-1}y_{k-1}+e_{k}+\Gamma _kx_N\nonumber\\
&=U_{k,k-1}x_{k-1}+(\Gamma _k-U_{k,k-1}\Gamma _{k-1})x_N+e_{k}\label{CML_Construct}
\end{align}
We consider $G_{k,k-1}=U_{k,k-1}$ and $G_{k,N}=\Gamma _k-U_{k,k-1}\Gamma _{k-1}$. Covariances $U_k$, $k \in [0,N-1]$, and $\text{Cov}(x_N)$ are nonsingular. So, covariances $G_k=\text{Cov}(e_k), k \in [0,N]$ (let $e_N=x_N$), are all nonsingular. So, $[x_k]$ is nonsingular (it can be shown similar to $\eqref{CML_Cinv}$). Thus, by $\eqref{CML_Construct}$, it can be seen that $[x_k]$ obeys $\eqref{CML_Dynamic_Forward}$--$\eqref{CML_Forward_BC2}$ (note that $[e_k]$ is white). So, $[x_k]$ is a ZMNG $CM_L$ sequence.

For $c=0$ we have a parallel proof. We skip the details and only present some results required later. Necessity: The proof is based on the $CM_F$ model. Let $[x_k]$ be a ZMNG $CM_F$ sequence governed by $\eqref{CML_Dynamic_Forward}$--$\eqref{CML_Forward_BC2}$ (for $c=0$). It is possible to represent $[x_k]$ as $\eqref{CML_Markov_z}$ with the Markov sequence $[y_k]_1^N$ governed by $y_k=U_{k,k-1}y_{k-1}+e_{k}$, $k \in [2,N]$, where for $k \in [2,N]$, $U_{k,k-1}=G_{k,k-1}$, $\Gamma _1=2G_{1,0}$, and $\Gamma _k=G_{k,k-1}\Gamma _{k-1}+G_{k,0}$.

Sufficiency: Let $[y_k]_1^N$ be a ZMNG Markov sequence governed by $y_k=U_{k,k-1}y_{k-1}+e_{k}$, $k \in [2,N]$, where $[e_{k}]_1^{N}$ (let $y_1=e_1$) is a zero-mean white NG sequence with covariances $U_{k}$. Also, let $x_0$ be a ZMNG vector uncorrelated with the sequence $[y_k]_1^N$. It can be shown that the sequence $[x_k]$ constructed by $\eqref{CML_Markov_z}$ obeys $\eqref{CML_Dynamic_Forward}$--$\eqref{CML_Forward_BC2}$ (for $c=0$), where for $k \in [2,N]$, $G_{k,k-1}=U_{k,k-1}$, $G_{1,0}=\frac{1}{2} \Gamma _1$, and $G_{k,0}=\Gamma _k-U_{k,k-1}\Gamma _{k-1}$.
\end{proof} 

Proposition \ref{CML_Markov_z_Proposition} presents a \textit{Markov-based representation} of $CM_c$ sequences. It reveals a key fact of relationship between a NG $CM_c$ sequence and a NG Markov sequence: the former is the latter plus an uncorrelated NG vector. As a result, it provides insight and guidelines for design of $CM_c$ models in application such as for motion trajectory modeling with destination information. A $CM_L$ model is more general than a reciprocal $CM_L$ model. Consequently, the following guidelines for $CM_L$ model design is more general than the approach of Theorem \ref{CML_R_Dynamic_FQ_Proposition}: First, consider a Markov model (e.g., a nearly constant velocity model) with the given origin density (without considering other information). The sequence so modeled is $[y_k]_0^{N-1}$ in $\eqref{CML_Markov_z}$. Assume the destination (density of $x_N$) is known. Then, based on $\Gamma _k$, the Markov sequence $[y_k]_0^{N-1}$ is modified to satisfy the available information in the problem (e.g., about the general form of trajectories) leading to the desired trajectories $[x_k]$ which end up at the destination. A direct attempt for designing parameters of a $CM_L$ model for this problem is much harder. These guidelines make parameter design easier and more intuitive. In addition, one can learn $\Gamma_k$ (which shows the impact of the destination) from a set of trajectories. Next, we study the representation of Proposition \ref{CML_Markov_z_Proposition} further to provide insight and tools for its application \cite{DD_Conf}.

The following representation of the $CM_c$ matrix is a by-product of Proposition \ref{CML_Markov_z_Proposition}.

\begin{corollary}\label{CML_Decomp}
Let $C$ be an $(N+1)d \times (N+1)d$ positive definite block matrix (with $(N+1)$ blocks in each row/column and each block being $d \times d$). $C^{-1}$ is $CM_c$ iff 
\begin{align}\label{CML_Decomp_equ}
C=B+\Gamma D\Gamma '
\end{align} 
Here $D$ is a $d\times d$ positive definite matrix and (i) if $c=N$, then $B=\left[\begin{array}{cc}
B_1 & 0\\
0 & 0
\end{array}\right]$ and $\Gamma =\left[\begin{array}{c}
S\\
I
\end{array}\right]$,
(ii) if $c=0$, then $B=\left[\begin{array}{cc}
0 & 0\\
0 & B_1
\end{array}\right]$ and $\Gamma =\left[\begin{array}{c}
I\\
S
\end{array}\right]$, where $(B_1)^{-1}$ is $Nd \times Nd$ block tri-diagonal, $S$ is $Nd \times d$, and $I$ is the $d\times d$ identity matrix.  

\end{corollary} 
\begin{proof}
Let $c=N$. Necessity: By Theorem \ref{CML_Characterization}, for every $CM_L$ matrix $C^{-1}$, there exists a ZMNG $CM_L$ sequence $[x_k]$ with the covariance $C$. By Proposition \ref{CML_Markov_z_Proposition}, we have 
\begin{align}\label{CML_zMarkov}
x=y+\Gamma x_N
\end{align}
where $x \triangleq [x_0' , x_1' , \ldots , x_N']$, $\mathsf{y} \triangleq [y_0' , y_1' , \ldots , y_{N-1}']'$, $y \triangleq [\mathsf{y}',0]'$,
$S \triangleq [\Gamma _0' , \Gamma _1' , \ldots , \Gamma _{N-1}']'$, $ \Gamma  \triangleq [S' , I]'$, and $[y_k]_0^{N-1}$ is a ZMNG Markov sequence uncorrelated with the ZMNG vector $x_N$. Then, by $\eqref{CML_zMarkov}$, we have
\begin{align}\label{CML_Cov}
\text{Cov}(x)=\text{Cov}(y)+\Gamma \text{Cov}(x_N)\Gamma '
\end{align}
because $y$ and $x_N$ are uncorrelated. Then, $\eqref{CML_Cov}$ leads to $\eqref{CML_Decomp_equ}$, where $C \triangleq \text{Cov}(x)$, $B\triangleq\left[\begin{array}{cc}
B_1 & 0 \\
0 & 0
\end{array}\right]=\text{Cov}(y)$, $B_1 \triangleq \text{Cov}(\mathsf{y})$, $D \triangleq \text{Cov}(x_N)$, and by Theorem \ref{CML_Characterization}, $(B_1)^{-1}$ is block tri-diagonal. Therefore, for every $CM_L$ matrix $C^{-1}$ we have $\eqref{CML_Decomp_equ}$. 

Sufficiency: Let $(B_1)^{-1}$ be an $Nd\times Nd$ block tri-diagonal matrix, $D$ be a $d\times d$ positive definite matrix, and $S$ be an $Nd\times d$ matrix. By Theorem \ref{CML_Characterization}, for every $Nd \times Nd$ block tri-diagonal matrix $(B_1)^{-1}$, there exists a Gaussian Markov sequence $[y_k]_0^{N-1}$ with $(C^\mathsf{y})^{-1}=(B_1)^{-1}$, where $C^\mathsf{y}=\text{Cov}(\mathsf{y})$ and $\mathsf{y}=[y_0',y_1',\ldots , y_{N-1}']'$. Also, given a $d\times d$ positive definite matrix $D$, there exists a Gaussian vector $x_N$ with $\text{Cov}(x_N)=D$. Let $x_N$ and $[y_k]_0^{N-1}$ be uncorrelated. By Proposition \ref{CML_Markov_z_Proposition}, $[x_k]$ constructed by $\eqref{CML_zMarkov}$ is a $CM_L$ sequence. Also, by Theorem \ref{CML_Characterization}, $C^{-1}$ of $[x_k]$ is a $CM_L$ matrix. With $C \triangleq \text{Cov}(x)$, $\eqref{CML_Decomp_equ}$ follows from $\eqref{CML_Cov}$. Thus, for every block tri-diagonal matrix $(B_1)^{-1}$, every positive definite matrix $D$, and every matrix $S$, $C^{-1}$ is a $CM_L$ matrix. The proof for $c=0$ is similar.
\end{proof}

\begin{corollary}\label{unique_rep}
For every $CM_c$ sequence, the representation $\eqref{CML_Markov_z}$ is unique.

\end{corollary}
\begin{proof}
Let $c=N$, and $[x_k]$ be a $CM_L$ sequence governed by $\eqref{CML_Dynamic_Forward}$ with parameters $(G_{k,k-1},G_{k,N},G_k)$, $k \in [1,N-1]$, and $\eqref{CML_Forward_BC2}$ with parameters $(G_{0,N},G_0,G_N)$. By Proposition \ref{CML_Markov_z_Proposition}, $[x_k]$ can be represented as $\eqref{CML_Markov_z}$. Parameters (denoted by $U_{k,k-1}, k \in [1,N-1]$, $U_{k}, k \in [0,N-1]$) of the Markov model $\eqref{Markov_Model}$ of $[y_k]_0^{N-1}$, covariance of $x_N$ denoted by $D$, and the matrices $\Gamma_k$, $k \in [0,N-1]$, can be calculated in terms of the parameters of the $CM_L$ model as follows (see the proof of Proposition \ref{CML_Markov_z_Proposition}):
\begin{align}
D&=G_N, \quad \Gamma _0=G_{0,N}\label{r_e1}\\
U_{k}&=G_k, \quad k \in [0,N-1]\label{r_e3}\\
U_{k,k-1}&=G_{k,k-1}, \quad k \in [1,N-1]\label{r_e4}\\
\Gamma _k&=G_{k,k-1}\Gamma _{k-1}+G_{k,N}, \quad k \in [1,N-1]\label{r_e6}
\end{align}
Now, assume that there exists a different representation of $[x_k]$ in the form $\eqref{CML_Markov_z}$. Denote parameters of the corresponding Markov model by $\tilde{U}_{k,k-1}, k \in [1,N-1]$, $\tilde{U}_{k}, k \in [0,N-1]$, and the weight matrices by $\tilde{\Gamma}_k$, $k \in [0,N-1]$ (covariance of $x_N$ is $D$). By the proof of Proposition \ref{CML_Markov_z_Proposition}, parameters of the corresponding $CM_L$ model are
\begin{align}
G_{0,N}&=\tilde{\Gamma}_0, \quad G_N=D\label{r_e1t}\\
G_{k,k-1}&=\tilde{U}_{k,k-1}, \quad k \in [1,N-1]\label{r_e4t}\\
G_{k,N}&=\tilde{\Gamma}_k-\tilde{U}_{k,k-1}\tilde{\Gamma}_{k-1}, \quad k \in  [1,N-1]\label{r_e5t}\\
G_k&=\tilde{U}_{k}, \quad k \in [0,N-1]\label{r_e6t}
\end{align}
Parameters of a $CM_L$ model of a $CM_L$ sequence are unique \cite{CM_Part_I_Conf}. Comparing $\eqref{r_e1}$-$\eqref{r_e6}$ and $\eqref{r_e1t}$-$\eqref{r_e6t}$, it can be seen that the parameters $\tilde{U}_{k,k-1}, k \in [1,N-1]$, $\tilde{U}_{k}, k \in [0,N-1]$, and $\tilde{\Gamma} _k$, $k \in [0,N-1]$, are the same as $U_{k,k-1}, k \in [1,N-1]$, $U_{k}, k \in [0,N-1]$, and $\Gamma _k$, $k \in [0,N-1]$. In other words, parameters of the representation $\eqref{CML_Markov_z}$ are unique. Uniqueness of $\eqref{CML_Markov_z}$ for $c=0$ can be proven similarly.
\end{proof}

Based on a valuable observation, \cite{ABRAHAM} discussed the relationship between Gaussian CM and Gaussian reciprocal processes. Then, based on the obtained relationship, \cite{ABRAHAM} presented a representation of NG reciprocal processes. We showed in \cite{CM_Part_II_A_Conf} that the relationship between Gaussian CM and Gaussian reciprocal processes presented in \cite{ABRAHAM} was incomplete; that is, the presented condition was not sufficient for a Gaussian process to be reciprocal (although \cite{ABRAHAM} stated that it was sufficient). Then, we presented in \cite{CM_Part_II_A_Conf} (Theorem \ref{CM_iff_Reciprocal} above) a relationship between CM and reciprocal processes for the general (Gaussian/non-Gaussian) case. Also, we showed that $CM_L$ in Theorem \ref{CM_iff_Reciprocal} was the missing part in the results of \cite{ABRAHAM}. Consequently, it can be seen that the representation presented in \cite{ABRAHAM} is not sufficient for a NG process to be reciprocal and its missing part is the representation of $CM_L$ processes.

Next, we present a simple necessary and sufficient representation of the NG reciprocal sequence from the CM viewpoint. It demonstrates the significance of studying reciprocal sequences from the CM viewpoint.

\begin{proposition}\label{Reciprocal_Markov_z}
A ZMNG $[x_k]$ is reciprocal iff it can be represented as both
\begin{align}
x_k&=y^L_k+\Gamma ^L_kx_N, \quad k \in [0,N-1] \label{CML_R}
\end{align}
and 
\begin{align}
x_k&=y^F_k+\Gamma ^F_kx_0, \quad  k \in [1,N]\label{CMF_R}
\end{align}
where $[y^L_k]_0^{N-1}$ and $[y^F_k]_1^N$ are ZMNG Markov sequences, $x_N$ and $x_0$ are ZMNG vectors uncorrelated with $[y^L_k]_0^{N-1}$ and $[y^F_k]_1^N$, respectively, and $\Gamma ^L_k$ and $\Gamma ^F_k$ are some matrices. 

\end{proposition}
\begin{proof}
A NG $[x_k]$ is reciprocal iff it is both $CM_L$ and $CM_F$ (Theorem \ref{CML_Characterization}). On the other hand, $[x_k]$ is $CM_L$ ($CM_F$) iff it can be represented as $\eqref{CML_R}$ ($\eqref{CMF_R}$) (Proposition \ref{CML_Markov_z_Proposition}). So, $[x_k]$ is reciprocal iff it can be represented as both $\eqref{CML_R}$ and $\eqref{CMF_R}$. 
\end{proof}

By $\eqref{CML_R}$--$\eqref{CMF_R}$ the relation between sample paths of the two Markov sequences is $y_k^L+\Gamma ^L_kx_N=y_k^F+\Gamma ^F_kx_0$, $k \in [1,N-1]$, $y^L_0+\Gamma ^L_0x_N=x_0$, $x_N=y^F_N + \Gamma ^F_Nx_0$.

The following representation of cyclic block tri-diagonal matrices is a by-product of Proposition \ref{Reciprocal_Markov_z}. 

\begin{corollary}\label{Reciprocal_Decomp}
Let $C$ be an $(N+1)d\times (N+1)d$ positive definite block matrix (with $(N+1)$ blocks in each row/column and each block being $d \times d$). Then, $C^{-1}$ is cyclic block tri-diagonal iff 
\begin{align}\label{Rec_Decomp}
C=B^L+\Gamma ^LD^L(\Gamma ^L)'=B^F+\Gamma ^FD^F(\Gamma ^F)'
\end{align} 
where $D^L$ and $D^F$ are $d\times d$ positive definite matrices, $B^L$ $=\left[\begin{array}{cc}
B_1 & 0\\
0 & 0
\end{array}\right]$, $\Gamma^L=\left[\begin{array}{c}
S_1\\
I
\end{array}\right]$, $B^F=\left[\begin{array}{cc}
0 & 0\\
0 & B_2
\end{array}\right]$, $\Gamma ^F=\left[\begin{array}{c}
I\\
S_2
\end{array}\right]$, $(B_1)^{-1}$ and $(B_2)^{-1}$ are $Nd \times Nd$ block tri-diagonal, $S_1$ and $S_2$ are $Nd \times d$, and $I$ is the $d\times d$ identity matrix.  

\end{corollary} 
\begin{proof}
Necessity: Let $C^{-1}$ be a positive definite cyclic block tri-diagonal matrix. So, $C^{-1}$ is $CM_L$ and $CM_F$. Then, by Corollary \ref{CML_Decomp} we have $\eqref{Rec_Decomp}$. Sufficiency: Let a positive definite matrix $C$ be written as $\eqref{Rec_Decomp}$. By Corollary \ref{CML_Decomp}, $C^{-1}$ is $CM_L$ and $CM_F$ and consequently cyclic block tri-diagonal. 
\end{proof}

The reciprocal sequence is an important special class of $CM_L$ ($CM_F$) sequences. So, it is important to know conditions for \eqref{CML_Markov_z} to represent a reciprocal sequence. 

\begin{proposition}\label{CML_z_Reciprocal}
Let $[y_k] \setminus \lbrace y_c \rbrace$, $c \in \lbrace 0,N \rbrace $, be a ZMNG Markov sequence, $y_k=U_{k,k-1}y_{k-1}+e_{k}, k \in [1,N] \setminus \lbrace a \rbrace$,
\begin{align*}
a&=\left\{ \begin{array}{cc}
1\,  & \text{if } c=0 \, \, \\
N & \text{if } c=N
\end{array} \right. ,  
\quad r=\left\{ \begin{array}{cc}
1\, \, \, \, \, \, \, & \text{if } c=0 \, \,  \\
N-1 & \text{if } c=N
\end{array} \right. 
\end{align*}
where $[e_k] \setminus \lbrace e_c \rbrace$ is a zero-mean white NG sequence with covariances $U_k$ (for $c=0$ we have $e_1=y_1$; for $c=N$ we have $e_0=y_0$). Also, let $x_c$ be a ZMNG vector with covariance $C_c$ uncorrelated with the Markov sequence $[y_k] \setminus \lbrace y_c \rbrace$. Let $[x_k]$ be constructed as
\begin{align}
x_k&=y_k+\Gamma _kx_c, \quad k \in [0,N] \setminus \lbrace c \rbrace \label{t1}
\end{align} 
where $\Gamma _k$ are some matrices. Then, $[x_k]$ is reciprocal iff $\forall k \in [1,N-1] \setminus \lbrace r \rbrace$,
\begin{align}
U_{k}^{-1}(\Gamma _k-U_{k,k-1}\Gamma _{k-1})=U_{k+1,k}'U_{k+1}^{-1}(\Gamma _{k+1}-U_{k+1,k}\Gamma _k) \label{CML_z_Condition_Reciprocal}
\end{align}
Moreover, $[x_k]$ is Markov iff in addition to $\eqref{CML_z_Condition_Reciprocal}$, we have
\begin{align}
&(U_0)^{-1}\Gamma _0=U_{1,0}'U_1^{-1}(\Gamma _1-U_{1,0}\Gamma _0), \, \, (\text{for} \, \, c=N) \label{CML_z_Markov}\\
&\Gamma _N-U_{N,N-1}\Gamma _{N-1}=0, \, \, (\text{for} \, \, c=0) \label{CMF_z_Markov}
\end{align}

\end{proposition}
\begin{proof}
 By Proposition \ref{CML_Markov_z_Proposition}, $[x_k]$ constructed by $\eqref{t1}$ is a $CM_c$ sequence. Parameters of the $CM_L$ model (i.e., $c=N$) are calculated by $\eqref{r_e1t}$--$\eqref{r_e6t}$ ($\tilde{U}_{k,k-1},k \in [1,N-1]$, $\tilde{U}_{k}, k \in [0,N-1]$ and $\tilde{\Gamma}_k, k \in [0,N-1]$ are replaced by $U_{k,k-1},k \in [1,N-1]$, $U_{k}, k \in [0,N-1]$, and $\Gamma _k, k \in [0,N-1]$). Parameters of the $CM_F$ model (i.e., $c=0$) are calculated as $G_{k,k-1}=U_{k,k-1}$, $k \in [2,N]$, $G_k=U_{k}$, $k \in [1,N]$, $G_{1,0}=\frac{1}{2} \Gamma _1$, $G_0=D$, $G_{k,0}=\Gamma _k-U_{k,k-1}\Gamma _{k-1}$, $ k \in [2,N]$. Then, by Proposition \ref{CML_R_Dynamic_Forward_Proposition}, the $CM_c$ sequence $[x_k]$ is reciprocal iff $\eqref{CML_z_Condition_Reciprocal}$ holds. Also, $[x_k]$ is Markov iff in addition to $\eqref{CML_z_Condition_Reciprocal}$, $\eqref{CML_z_Markov}$ holds for $c=N$ or $\eqref{CMF_z_Markov}$ for $c=0$.
\end{proof}

Due to their importance in designing $CM_c$ dynamic models, the main elements of representation $\eqref{CML_Markov_z}$ are formally defined.

\begin{definition}\label{CML_Constructed}
In $\eqref{CML_Markov_z}$, $[y_k] \setminus \lbrace y_c \rbrace$ is called an \textit{underlying} Markov sequence and the Markov model (without the initial condition) is called an \textit{underlying} Markov model. Also, $[x_k]$ is called a $CM_c$ sequence \textit{constructed} from an underlying Markov sequence and the $CM_c$ model (without the boundary condition) is called a $CM_c$ model \textit{constructed} from an underlying Markov model. 

\end{definition}

\begin{corollary}\label{CML_Markov_Model}
For $CM_c$ models, having the same underlying Markov model is equivalent to having the same $G_{k,k-1}, G_k$, $\forall k \in [1,N] \setminus \lbrace a \rbrace$ ($a=N$ for $c=N$, and $a=1$ for $c=0$).

\end{corollary}
\begin{proof}
Given a Markov model with parameters $U_{k,k-1}, U_k, k $ $ \in [1,N] \setminus \lbrace a \rbrace$, by our proof of Proposition \ref{CML_Markov_z_Proposition}, parameters of a $CM_c$ model constructed from the Markov model are $G_{k,k-1}=U_{k,k-1}$, $G_{k,c}=\Gamma _k-U_{k,k-1}\Gamma _{k-1}, G_k=U_{k}, k \in [1,N] \setminus \lbrace a \rbrace$. Clearly all $CM_c$ models so constructed have the same $G_{k,k-1}, G_{k}, k \in [1,N] \setminus \lbrace a \rbrace$.

For a $CM_c$ model with parameters $G_{k,k-1},G_{k,c},G_k$, $\forall k \in [1,N] \setminus \lbrace a \rbrace$, parameters of its underlying Markov model are uniquely determined as (see the proof of Proposition \ref{CML_Markov_z_Proposition})
\begin{align}
U_{k,k-1}&=G_{k,k-1}, \, \, \, \, U_{k}=G_k, \, \, \, \, k \in [1,N] \setminus \lbrace a \rbrace \label{M_1}
\end{align} 
So, $CM_c$ models with the same $G_{k,k-1},G_k$, $\forall k \in [1,N] \setminus \lbrace a \rbrace$, are constructed from the same underlying Markov model.
\end{proof}

In the following, we try to distinguish between two concepts which are both useful in the application of $CM_L$ and reciprocal sequences: 1) a $CM_L$ model \textit{induced} by a Markov model (Definition \ref{CML_Derived}) and 2) a $CM_L$ model \textit{constructed} from its underlying Markov model (Definition \ref{CML_Constructed}). 

By Theorem \ref{CML_R_Dynamic_FQ_Proposition}, a Markov-induced $CM_L$ model is actually a reciprocal $CM_L$ model. In other words, non-reciprocal $CM_L$ models can not be so induced (with $\eqref{CML_Choice_1}$--$\eqref{CML_Choice_3}$) by any Markov model. By Corollary \ref{Reciprocal_Derived_Markov}, every reciprocal $CM_L$ model can be so induced. However, the corresponding Markov model is not unique. In addition, every Markov sequence modeled by a Markov model is also modeled by the corresponding Markov-induced $CM_L$ model (with an appropriate boundary condition).    

Every $CM_L$ model can be constructed from its underlying Markov model, and the construction is unique (Corollary \ref{unique_rep}). So, an underlying Markov model plays a fundamental role in constructing a $CM_L$ model. However, an underlying Markov sequence is not modeled by the constructed $CM_L$ model. 

The underlying Markov model of a Markov-induced $CM_L$ model can be determined as follows. Let $M_{k,k-1},M_{k}$, $\forall k \in [1,N]$, be the parameters of a Markov model $\eqref{Markov_Model}$. Parameters of its induced $CM_L$ model are calculated by $\eqref{CML_Choice_1}$--$\eqref{CML_Choice_3}$. Then, by $\eqref{M_1}$, parameters of the underlying Markov model are, $ \forall k \in [1,N-1]$,
\begin{align}
U_{k,k-1}&=M_{k,k-1}-(U_{k}M_{N|k}'C_{N|k}^{-1})M_{N|k-1}\label{M_U_1}\\
U_{k}&=(M_{k}^{-1}+M_{N|k}'C_{N|k}^{-1}M_{N|k})^{-1}\label{M_U_2}
\end{align}
$M_{N|k}=M_{N,N-1}\cdots M_{k+1,k}$, $k \in [0,N-1]$, $C_{N|k}=\sum _{n=k}^{N-1} $ $M_{N|n+1}M_{n+1}M_{N|n+1}'$, $k \in [1,N-1]$, and $M_{N|N}=I$.

\color{black}

\section{Summary and Conclusions}\label{Summary_Conclusions}

Dynamic models for different classes of nonsingular Gaussian (NG) conditionally Markov (CM) sequences have been studied, and approaches/guidelines for designing their parameters have been presented. Every reciprocal $CM_L$ model can be induced by a Markov model. More specifically, parameters of the former can be obtained from those of the latter. In some applications sequences belonging to more than one CM class are desired. It has been shown how dynamic models of such NG CM sequences can be obtained. A spectrum of CM dynamic models has been presented, which makes a gradual change of models from a $CM_L$ model to a reciprocal $CM_L$ model clear. A NG $CM_c$ sequence can be represented by a sum of a NG Markov sequence and an uncorrelated NG vector. It is useful for designing $CM_L$ and $CM_F$ models. Moreover, a representation of NG reciprocal sequences has been presented from the CM viewpoint, which demonstrates the significance of studying reciprocal sequences from the CM viewpoint. Our results here provide some theoretical tools for application of CM sequences, e.g., trajectory modeling and prediction with destination information. 

Equivalent CM dynamic models were studied in \cite{CM_Conf_Explicitly}--\cite{CM_Journal_Algebraically}. Also, general singular/nonsingular Gaussian CM sequences and their application were studied in \cite{CM_SingularNonsingular}--\cite{Thesis_Reza}.


\begin{thebibliography}{99}
\bibitem{CM_Part_I_Conf} R. Rezaie and X. R. Li. Nonsingular Gaussian Conditionally Markov Sequences. \textit{IEEE West. New York Image and Signal Processing Workshop}. Rochester, NY, USA, Oct. 2018, pp. 1-5.  
\bibitem{Fanas1} M. Fanaswala and V. Krishnamurthy. Detection of Anomalous Trajectory Patterns in Target Tracking via Stochastic Context-Free Grammar and Reciprocal Process Models. \textit{IEEE J. of Selected Topics in Signal Processing}, Vol. 7, No. 1, pp. 76-90, 2013.
\bibitem{Fanas2} M. Fanaswala, V. Krishnamurthy, and L. B. White. Destination-aware Target Tracking via Syntactic Signal Processing. \textit{IEEE Inter. Conference on Acoustics, Speech and Signal Processing}, Prague, Czech, May 2011.
\bibitem{White_Waypoint} L. B. White and F. Carravetta. Normalized Optimal Smoothers for a Class of Hidden Generalized Reciprocal Processes. \textit{IEEE Trans. on Automatic Control}, Vol. 62, No. 12, pp. 6489-6496, 2017.
\bibitem{White_Tracking1} G. Stamatescu, L. B. White , and R. B. Doust. Track Extraction With Hidden Reciprocal Chains. \textit{IEEE Trans. on Automatic Control}. Vol. 63, No. 4, pp. 1097-1104, 2018.
\bibitem{Simon} B. I. Ahmad, J. K. Murphy, S. J. Godsill, P. M. Langdon, and R. Hardy. Intelligent Interactive Displays in Vehicles with Intent Prediction: A Bayesian Framework. \textit{IEEE Signal Processing Magazine}, Vol. 34, No. 2, pp. 82-94, 2017.
\bibitem{Simon2} B. I. Ahmad, J. K. Murphy, P. M. Langdon, and S. J. Godsill. Bayesian Intent Prediction in Object Tracking Using Bridging Distributions. \textit{IEEE Trans. on Cybernetics}, Vol. 48, No. 1, pp. 215-227, 2018.
\bibitem{Krener1} A. J. Krener. Reciprocal Processes and the Stochastic Realization Problem for Acausal Systems. \textit{Modeling, Identification, and Robust Control}, C. I. Byrnes and A. Lindquist (editors), Elsevier, 1986.
\bibitem{Picci} A. Chiuso, A. Ferrante, and G. Picci. Reciprocal Realization and Modeling of Textured Images. \textit{44th IEEE Conference on Decision and Control}, Seville, Spain, Dec. 2005. 
\bibitem{Picci2} G. Picci and F. Carli. Modelling and Simulation of Images by Reciprocal Processes. \textit{Tenth Inter. Conference on Computer Modeling and Simulation}, Cambridge, UK, Apr. 2008.
\bibitem{DD_Conf} R. Rezaie and X. R. Li. Destination-Directed Trajectory Modeling and Prediction Using Conditionally Markov Sequences. \textit{IEEE Western New York Image and Signal Processing Workshop}, Rochester, NY, USA, Oct. 2018, pp. 1-5.
\bibitem{DW_Conf} R. Rezaie and X. R. Li. Trajectory Modeling and Prediction with Waypoint Information Using a Conditionally Markov Sequence. \textit{56th Allerton Conference on Communication, Control, and Computing}, Monticello, IL, USA, Oct. 2018, pp. 486-493.
\bibitem{Bernstein} S. Bernstein. Sur Les Liaisons Entre Les Grandeurs Aleatoires. \textit{Verhand. Internat. Math. Kongr.}, Zurich, (Band I), 1932.
\bibitem{Schrodinger_1} E. Schrodinger. Uber die Umkehrung der Naturgesetze. \textit{Sitz. Ber. der Preuss. Akad. Wissen., Berlin Phys. Math.}, 144, 1931.
\bibitem{Levy_2} B. Levy and A. J. Krener. Stochastic Mechanics of Reciprocal Diffusions. \textit{J. of Mathematical Physics}. Vol. 37, No. 2, pp. 769-802, 1996.
\bibitem{Mehr} C. B. Mehr and J.  A. McFadden. Certain Properties of Gaussian Processes and their First-Passage Times. \textit{J. of Royal Statistical Society}, (B), Vol. 27, pp. 505-522, 1965.
\bibitem{ABRAHAM} J. Abraham and J. Thomas. Some Comments on Conditionally Markov and Reciprocal Gaussian Processes. \textit{IEEE Trans. on Information Theory}, Vol. 27, No. 4, pp. 523-525, 1981.
\bibitem{CM_Part_II_A_Conf} R. Rezaie and X. R. Li. Gaussian Reciprocal Sequences from the Viewpoint of Conditionally Markov Sequences. \textit{Inter. Conference on Vision, Image and Signal Processing}, Las Vegas, NV, USA, Aug. 2018., pp. 33:1-33:6.
\bibitem{Jamison_Reciprocal} B. Jamison. Reciprocal Processes. \textit{Z. Wahrscheinlichkeitstheorie verw. Gebiete}, Vol. 30, pp. 65-86, 1974.  
\bibitem{Reciprocal_Measure} C. Leonard, S. Rœlly, and J-C Zambrini. Reciprocal Processes. A Measure-theoretical Point of View. \textit{Probability Surveys}, Vol. 11, pp. 237–269, 2014.
\bibitem{Conforti} G. Conforti, P. Dai Pra, and S. Roelly. Reciprocal Class of Jump Processes. \textit{J. of Theoretical Probability}, vol. 30, no. 2, pp. 551-580, 2017.
\bibitem{Murr} R. Murr, \textit{Reciprocal Classes of Markov Processes. An Approach with Duality Formulae}, Ph.D. Thesis, Universitat Potsdam, 2012.
\bibitem{Roally} S. Rœlly. \textit{Reciprocal Processes. A Stochastic Analysis Approach}. In V. Korolyuk, N. Limnios, Y. Mishura, L. Sakhno, and G. Shevchenko, editors, Modern Stochastics and Applications, volume 90 of Optimization and Its Applications, pp. 53–67, Springer, 2014. 
\bibitem{Carm} J-P Carmichael, J-C Masse, and R. Theodorescu. Representations for Multivariate Reciprocal Gaussian Processes. \textit{IEEE Trans. on Information Theory}, Vol. 34, No. 1, pp. 155-157, 1988.
\bibitem{Carm2} J-P Carmichael, J-C Masse, and R. Theodorescu. Multivariate Reciprocal Stationary Gaussian Processes. \textit{J. of Multivariate Analysis}, Vol. 23, pp. 47-66, 1987.
\bibitem{Krener_2} A. J. Krener, R. Frezza, and B. C. Levy. Gaussian Reciprocal Processes and Self-adjoint Stochastic Differential Equations of Second Order. \textit{Stochastics and Stochastic Reports}, Vol. 34, No. 1-2, pp. 29–56, 1991.
\bibitem{Levy_Class} B. C. Levy and A. Beghi, Discrete-time Gauss-Markov Processes with Fixed Reciprocal Dynamics. \textit{J. of Mathematical Systems, Estimation, and Control}, Vol. 4, No. 3, pp. 1-25, 1994.
\bibitem{Levy_Dynamic} B. C. Levy, R. Frezza, and A. J. Krener. Modeling and Estimation of Discrete-Time Gaussian Reciprocal Processes. \textit{IEEE Trans. on Automatic Control}, Vol. 35, No. 9, pp. 1013-1023, 1990.
\bibitem{Bacca1} E. Baccarelli and R. Cusani. Recursive Filtering and Smoothing for Gaussian Reciprocal Processes with Dirichlet Boundary Conditions. \textit{IEEE Trans. on Signal Processing}, Vol. 46, No. 3, pp. 790-795, 1998.
\bibitem{Bacca2} E. Baccarelli, R. Cusani, and G. Di Blasio. Recursive Filtering and Smoothing for Reciprocal Gaussian Processes-pinned Boundary Case. \textit{IEEE Trans. on Information Theory}, Vol. 41, No. 1, pp. 334-337, 1995.
\bibitem{Moura1} D. Vats and J. M. F. Moura, Recursive Filtering and Smoothing for Discrete Index Gaussian Reciprocal Processes. \textit{43rd Annual Conference on Information Sciences and Systems}, Baltimore, USA, Mar. 2009.
\bibitem{Moura2} D. Vats and J. M. F. Moura. Telescoping Recursive Representations and Estimation of Gauss–Markov Random Fields. \textit{IEEE Trans. on Information Theory}, Vol. 57, No. 3, pp. 1645-1663, 2011.
\bibitem{Carli} F. P. Carli, A. Ferrante, M. Pavon, and G. Picci. A Maximum Entropy Solution of the Covariance Extension Problem for Reciprocal Processes. \textit{IEEE Trans. on Automatic Control}, Vol. 56, No. 9, pp. 1999-2012, 2011.
\bibitem{White} F. Carravetta and L. B. White. Modeling and Estimation for Finite-State Reciprocal Processes. \textit{IEEE Trans. on Automatic Control}, Vol. 57, No. 9, pp. 2190-2202, 2012.
\bibitem{Carra} F. Carravetta. Nearest-Neighbor Modeling of Reciprocal Chains. \textit{An Inter. J. of Probability and Stochastic Processes}, Vol. 80, No. 6, pp. 525-584, 2008. 
\bibitem{White_2} L. B. White and F. Carravetta. Optimal Smoothing for Finite-State Hidden Reciprocal Processes. \textit{IEEE Trans. on Automatic Control}, Vol. 56, No. 9, pp. 2156-2161, 2011.
\bibitem{White_3} L B. White and H. X. Vu. Maximum Likelihood Sequence Estimation for Hidden Reciprocal Processes. \textit{IEEE Trans. on Automatic Control}, Vol. 58, No. 10, pp. 2670-2674, 2013. 
\bibitem{Hwang0} I. Hwang and C. E. Seah. Intent-Based Probabilistic Conflict Detection for the Next Generation Air Transportation System. \textit{ Proceedings of the IEEE}, Vol. 96, No 12, pp. 2040-2058, 2008.
\bibitem{Hwang1} J. Yepes, I. Hwang, and M. Rotea. An Intent-Based Trajectory Prediction Algorithm for Air Traffic Control. \textit{AIAA Guidance, Navigation, and Control Conference}, San Francisco, CA, Aug. 2005.
\bibitem{Hwang2}  J. Yepes, I. Hwang, and M. Rotea. New Algorithms for Aircraft Intent Inference and Trajectory Prediction. \textit{AIAA Journal of Guidance, Control, and Dynamics}, Vol. 30, No. 2, pp. 370-382, 2007.  
\bibitem{Louis} Y. Liu and X. R. Li. Intent-Based Trajectory Prediction by Multiple Model Prediction and Smoothing. \textit{AIAA Guidance, Navigation, and Control Conference}, Kissimmee, Florida, 2015. 
\bibitem{Krozel} J. Krozel and D. Andrisani. Intent Inference and Strategic Path Prediction. \textit{AIAA Guidance, Navigation, and Control Conference}, San Francisco, CA, Aug. 2005.
\bibitem{CM_Part_II_B_Conf} R. Rezaie and X. R. Li. Models and Representations of Gaussian Reciprocal and Conditionally Markov Sequences. \textit{Inter. Conference on Vision, Image and Signal Processing}, Las Vegas, NV, USA, Aug. 2018, pp. 65:1-65:6.
\bibitem{Ackner} R. Ackner and T. Kailath. Discrete-Time Complementary Models and Smoothing. \textit{Inter. J. of Control}, Vol. 49, No. 5, pp. 1665-1682, 1989.
\bibitem{Gaussian_2} A. L. Barker, D. E. Brown, and W. N. Martin. Bayesian Estimation and the Kalman Filter. \textit{Computers Mathematics Application}, Vol. 30, No 10, pp. 55-77, 1995.


\bibitem{CM_Conf_Explicitly} R. Rezaie and X. R. Li. Explicitly Sample-Equivalent Dynamic Models for Gaussian Conditionally Markov, Reciprocal, and Markov Sequences. \textit{Inter. Conf. on Control, Automation, Robotics, and Vision Engineering}, New Orleans, LA, USA, Nov. 2018, pp. 1-6.
\bibitem{CM_Journal_Algebraically} R. Rezaie and X. R. Li. Gaussian Conditionally Markov Sequences: Algebraically Equivalent Dynamic Models. \textit{IEEE T-AES}, 2019, DOI: 10.1109/TAES.2019.2951188.

\bibitem{CM_SingularNonsingular} R. Rezaie and X. R. Li. Gaussian Conditionally Markov Sequences: Singular/Nonsingular. \textit{IEEE T-AC}, 2019, DOI: 10.1109/TAC.2019.2944363.
\bibitem{Thesis_Reza} R. Rezaie. \textit{Gaussian Conditionally Markov Sequences: Theory with Application}. Ph.D. Dissertation, Dept of Electrical Engineering, University of New Orleans, July 2019.


\end{thebibliography}
\end{document}